\documentclass[journal]{IEEEtran}

\usepackage{amsfonts}
\usepackage{color}
\usepackage{graphicx}
\usepackage[dvips]{epsfig}
\usepackage{graphics} 
\usepackage{times} 
\usepackage[cmex10]{amsmath} 
\usepackage{amssymb}  
\usepackage{cite}
\usepackage{multirow}
\usepackage[tight,footnotesize]{subfigure}
\usepackage{amsmath}
\usepackage[boxed,ruled,lined]{algorithm2e}

\def\norm #1{\left\|#1\right\|}
\def\inftyn #1{\left\|#1\right\|_{\infty}}
\def\twon #1{\left\|#1\right\|_2}
\def\onen #1{\left\|#1\right\|_1}
\def\zeron #1{\left\|#1\right\|_0}

\def\abs #1{\left|#1\right|}

\def\bR{\mathbb{R}}

\def\bS{\mathbb{S}}

\def\m #1{\boldsymbol{#1}}

\def\cB{\mathcal{B}}

\def\cK{\mathcal{K}}
\def\cL{\mathcal{L}}
\def\cN{\mathcal{N}}
\def\cP{\mathcal{P}}

\def\bee{\begin{equation}}
\def\ene{\end{equation}}

\def\beq{\begin{eqnarray}}
\def\enq{\end{eqnarray}}
\def\lentwo{\setlength\arraycolsep{2pt}}

\newtheorem{lem}{Lemma}
\newtheorem{rem}{Remark}

\newtheorem{thm}{Theorem}
\newtheorem{prop}{Proposition}

\def\equ #1{\begin{equation}#1\end{equation}}
\def\equa #1{\begin{eqnarray}#1\end{eqnarray}}
\def\sbra #1{\left(#1\right)}
\def\mbra #1{\left[#1\right]}
\def\lbra #1{\left\{#1\right\}}
\def\diag #1{\text{diag}#1}

\title{Bayesian Compressed Sensing With New Sparsity-Inducing Prior}
\author{Zai Yang, Lihua Xie$^*$, and Cishen Zhang

\thanks{$^*$Author for correspondence.

Z. Yang and L. Xie are with EXQUISITUS, Centre for E-City, School of Electrical and Electronic Engineering, Nanyang Technological University, 639798, Singapore (e-mail: yang0248@e.ntu.edu.sg; elhxie@ntu.edu.sg).

C. Zhang is with the Faculty of Engineering and Industrial Sciences, Swinburne University of Technology, Hawthorn VIC 3122, Australia (e-mail: cishenzhang@swin.edu.au).}}

\begin{document}
\maketitle

\begin{abstract} Sparse Bayesian learning (SBL) is a popular approach to sparse signal recovery in compressed sensing (CS). In SBL, the signal sparsity information is exploited by assuming a sparsity-inducing prior for the signal that is then estimated using Bayesian inference. In this paper, a new sparsity-inducing prior is introduced and efficient algorithms are developed for signal recovery. The main algorithm is shown to produce a sparser solution than existing SBL methods while preserving their desirable properties. Numerical simulations with one-dimensional synthetic signals and two-dimensional images verify our analysis and show that for sparse signals the proposed algorithm outperforms its SBL peers in both the signal recovery accuracy and computational speed. Its improved performance is also demonstrated in comparison with other state-of-the-art methods in CS.
\end{abstract}

\begin{IEEEkeywords}
Compressed sensing, G-STG prior, greedy algorithm, sparse Bayesian learning.
\end{IEEEkeywords}


\section{Introduction}
In practice, one would like to determine a high-dimensional signal $\m{x}\in\bR^N$ from its relatively low-dimensional linear measurements $\m{y}=\m{A}\m{x}\in\bR^M$, where $\m{A}\in\bR^{M\times N}$ and $M<N$. There exist infinite candidates of $\m{x}$ that satisfy the linear equation in general. To possibly recover the true signal, some additional information needs to be taken into account. Fortunately, most signals of interest are sparse or compressible under appropriate bases, e.g., an image under a wavelet basis. Without loss of generality, we consider a signal that is sparse under the canonical basis since any sparsifying transform of $\m{x}$ can be absorbed into the matrix $\m{A}$. So we need to search for the maximally sparse solution to the linear equation, i.e., to minimize $\zeron{\m{x}}$ that counts the number of nonzero entries of $\m{x}$ subject to the constraint $\m{A}\m{x}=\m{y}$. Since this combinatorial optimization problem is NP-hard, its convex relaxation (replacing $\zeron{\m{x}}$ by $\onen{\m{x}}$) coined as basis pursuit (BP) has been extensively studied (see \cite{tropp2006just} and references therein). During the past several years, the sparse signal recovery problem has been developed into a well-known research area named as compressed sensing (CS)\cite{candes2006compressive}, which has wide applications including medical imaging \cite{lustig2007sparse}, source localization \cite{yang2012robustly} and single-pixel camera \cite{duarte2008single}, to name just a few. In CS, one wishes to recover a sparse signal $\m{x}$ from its compressive measurements $\m{y}=\m{A}\m{x}$. The sensing matrix $\m{A}$ is typically generated from a random distribution such that it satisfies some desirable properties such as restricted isometry property (RIP) \cite{candes2006compressive}. It is shown that $\ell_1$ optimization can recover the true signal exactly under mild conditions. In addition, $\ell_1$ optimization is robust to measurement noises and works efficiently with compressible signals \cite{candes2006stable}. While the $\ell_1$ norm used to promote sparsity is a convex approximation to the original $\ell_0$ norm, nonconvex optimizations have also been studied. It is shown in \cite{chartrand2007exact,chartrand2008restricted} that improved performance can be obtained using the $\ell_p$ ($0<p<1$) norm. The nonconvex objective function $\sum_{i=1}^N\log\sbra{\abs{x_i}+\tau}$ with $\tau\geq0$ is related to reweighted $\ell_1$ minimization in \cite{candes2008enhancing}. Another class of approaches to CS uses a greedy pursuit method including OMP \cite{tropp2007signal} and StOMP\cite{donoho2006sparse}. In a greedy algorithm, the support of the solution is modified sequentially with local optimization in each step.

This paper is focused on Bayesian approaches to CS, known as Bayesian CS \cite{ji2008bayesian}. This method was originated from the area of machine leaning and introduced by Tipping \cite{tipping2001sparse} for obtaining sparse solutions to regression and classification tasks that use models which are linear in the parameters, coined as relevance vector machine (RVM) or sparse Bayesian learning (SBL). SBL is built upon a statistical perspective where the sparsity information is exploited by assuming a sparsity-inducing prior for the signal of interest that is then estimated via Bayesian inference. Its theoretical performance is analyzed by Wipf and Rao \cite{wipf2004sparse}. After being introduced into CS by Ji \emph{et al.} \cite{ji2008bayesian}, this technique has become a popular approach to CS and other sparsity-related problems. For example, a Bayesian framework is presented in \cite{wipf2009unified} for MEG/EEG source imaging. In \cite{yang2012unified}, the authors introduce a framework to unify the CS problems with multi- and one-bit quantized measurements and estimate the sparse signals of interest and quantization errors based on a Bayesian formulation. Currently, main research topics in SBL for CS include (a) developing efficient sparsity-inducing priors, (b) incorporating additional signal structures in the prior besides sparsity and (c) designing fast and accurate inference algorithms. This paper studies problems (a) and (c). Examples of (b) include \cite{he2009exploiting} that exploits the wavelet structure of images and \cite{yu2011bayesian} on cluster structured sparse signals.



Many sparsity-inducing priors have been studied in the literature. In \cite{he2009exploiting,yu2011bayesian}, a spike-and-slab prior \cite{ishwaran2005spike} is applied which is a mixture of a point mass at zero and a continuous distribution elsewhere and fits naturally for sparse signals. A typical inference scheme for such a prior is a Markov chain Monte Carlo (MCMC) method \cite{robert2004monte} due to the lack of closed-form expressions of Bayesian estimators. As a result, the inference process may suffer from computational difficulties because a large number of samples are required to approximate the posterior distribution and the convergence is typically slow. A popular class of sparsity-inducing priors is introduced in a hierarchical framework where a complex prior is composed of two or more simple distributions. For example, a Student's $t$-prior (or Gaussian-inverse gamma prior) is used in the basic SBL \cite{tipping2001sparse} that is composed of a Gaussian prior in the first layer and a gamma prior in the second. A Laplace (Gaussian-exponential) prior is used in \cite{babacan2010bayesian}. A Gaussian-gamma prior is recently studied in \cite{pedersen2011sparse} that generalizes the Laplace prior. Two popular inference methods for the hierarchical priors are evidence procedure \cite{mackay1992bayesian}, e.g., in \cite{tipping2001sparse,babacan2010bayesian}, and variational Bayesian inference \cite{beal2003variational}, e.g., in \cite{shutin2011sparse}. Both the methods are approximations of Bayesian inference since the exact inference is intractable. In an evidence procedure, the signal estimator has a simple expression in which some unknown hyperparameters are involved and estimated iteratively by maximizing their evidence. In a variational Bayesian inference method, the posterior distribution is approximated using some family of tractable distributions followed by computation of an optimal distribution within the family. To circumvent high-dimensional matrix inversions, a fast algorithm framework is developed in \cite{tipping2003fast} for evidence procedure and also adopted in \cite{babacan2010bayesian}. But it is observed in this paper and \cite{pedersen2011sparse} that the fast algorithms in \cite{tipping2003fast,babacan2010bayesian} typically produce signal estimates with overestimated support sizes, especially in a low signal to noise ratio (SNR) regime or in the case of a large sample size. A new sparsity-inducing prior and algorithm are proposed in this work to resolve this problem with improved convergence speed.

Though formulated from a different perspective, SBL is related to other approaches to CS. Consider the observation model $\m{y}=\m{A}\m{x}+\m{e}$ where $\m{e}$ represents an additive white Gaussian noise (AWGN). Let $p\sbra{\m{x}}$ be the prior for $\m{x}$. Then, a maximum {\em a posteriori} (MAP) estimator of $\m{x}$ coincides with a solution to a regularized least-squares problem with $-\log p\sbra{\m{x}}$ (up to a scale) being the regularization term, which bridges SBL and optimization methods. For example, a Laplace prior corresponds to the widely studied $\ell_1$ minimization. A prior corresponding to the nonconvex $\ell_p$ ($0<p<1$) norm is studied in \cite{babacan2009non}. The fast algorithm in \cite{tipping2003fast} is related to the greedy pursuit method. In fact, it is a greedy algorithm using a different support modification criterion. Unlike OMP and StOMP, it allows deletion of irrelevant basis vectors that may have been added to the solution support in earlier steps.

In this paper, we introduce a new sparsity-inducing prior named as Gaussian shifted-truncated-gamma (G-STG) prior that generalizes the Gaussian-gamma prior in \cite{pedersen2011sparse}. The extended flexibility of the new prior promotes its capability of modeling sparse signals. In fact, it is shown that the Gaussian-gamma prior cannot work in the main algorithm of this paper. From the perspective of MAP estimation, the G-STG prior corresponds to a nonconvex objective function in optimization methods in general. For signal recovery we propose an iterative algorithm based on an evidence procedure and a fast greedy algorithm inspired by the algorithm in \cite{tipping2003fast}. We show that similar theoretical guarantees shown in \cite{wipf2004sparse} hold for the new SBL method as for the basic SBL. Specifically, we show that every local optimum of the SBL cost function is achieved at a sparse solution and that the global optimum is achieved at the maximally sparse solution. Moreover, we show that the proposed algorithm produces a sparser solution than existing SBL methods. We provide simulation results with one-dimensional synthetic signals and two-dimensional images that verify our analysis. We compare our proposed method with state-of-the-art ones to illustrate its improved performance.



Notations used in this paper are as follows. Bold-face letters are reserved for vectors and matrices. For ease of exposition, we do not distinguish a random variable from its numerical value. $x_i$ is the $i$th entry of a vector $\m{x}$. $\zeron{\m{x}}$ counts the number of nonzero entries of $\m{x}$. $\norm{\m{x}}_p=\sbra{\sum_i \abs{x_i}^p}^{1/p}$ for $p>0$ denotes the $\ell_p$ norm (or pseudo-norm) of $\m{x}$. $\diag\sbra{\m{x}}$ denotes a diagonal matrix with diagonal entries being the elements of the vector $\m{x}$. $A_{ij}$ denotes the $(i,j)$th entry of a matrix $\m{A}$. $\abs{\m{A}}$ denotes the determinant of the matrix $\m{A}$. $E\lbra{v}$ denotes the expectation of a random variable $v$.

The rest of the paper is organized as follows. Section \ref{sec:newsparseprior} introduces the G-STG prior. Section \ref{sec:algorithms} presents the Bayesian framework, an iterative procedure for signal recovery, and theoretical results of the new SBL method. Section \ref{sec:greedyalg} introduces a fast greedy algorithm with analysis. Section \ref{sec:simulation} provides empirical results to show the efficiency of the proposed solution. Section \ref{sec:conclusion} concludes this paper.

\section{G-STG Prior} \label{sec:newsparseprior}
\subsection{Mathematical Formulation}
We introduce the hierarchical Gaussian shifted-truncated-gamma (G-STG) prior for a sparse signal $\m{x}\in\bR^N$ as follows:
{\lentwo\equa{p\sbra{\m{x}|\m{\alpha}}
&=&\cN\sbra{\m{x}|\m{0},\m{\Lambda}},\label{formu:prior_x}\\p\sbra{\m{\alpha};\tau,\epsilon,\eta}
&=&\prod_{i=1}^Np\sbra{\alpha_i;\tau,\epsilon,\eta}=\prod_{i=1}^{N} \Gamma_{\tau}\sbra{\alpha_i|\epsilon,\eta},\label{formu:prior_alpha}}
}where $\m{\Lambda}=\diag\sbra{\m{\alpha}}$, $\m{\alpha}\in\bR^{N}$, $p\sbra{\alpha_i;\tau,\epsilon,\eta}$ is a shifted-truncated-gamma (STG) distribution for $\alpha_i\geq0$ with
\equ{\Gamma_{\tau}\sbra{\alpha_i|\epsilon,\eta}=\frac{\eta^{\epsilon}} {\Gamma_{\eta\tau}\sbra{\epsilon}} \sbra{\alpha_i+\tau}^{\epsilon-1}\exp\lbra{-\eta\sbra{\alpha_i+\tau}}, \label{formu:pdf_trun_gamma}}
$\epsilon\geq0$ is the shape parameter, $\eta\geq0$ is the rate parameter, $\tau\geq0$ is the threshold parameter and $\Gamma_{\tau}\sbra{\epsilon}=\int_{\tau}^{\infty}t^{\epsilon-1}e^{-t}dt$ denotes an incomplete gamma function. The first layer of the prior is a commonly used Gaussian prior that leads to convenient computations as shown later. In the second layer $\alpha_i$, $i=1,\cdots,N$, are assumed to be independent, and further $\alpha_i+\tau$, $i=1,\cdots,N$, are i.i.d. truncated gamma distribution (that is why we say that $p\sbra{\alpha_i;\tau,\epsilon,\eta}$ is an STG distribution). By $p\sbra{\alpha_i}\propto\sbra{\alpha_i+\tau}^{\epsilon-1} \exp\lbra{-\eta\alpha_i}$, $i=1,\cdots,N$, it is obvious that $\alpha_i$ is favored to be zero in the second layer if $\epsilon\leq1$, resulting in that $x_i$ is favored to be zero. Thus the hierarchical prior is a sparsity-inducing prior. In general, there is no explicit expression for the marginal distribution
\equ{\begin{split}p\sbra{\m{x};\tau,\epsilon,\eta}
&=\prod_{i=1}^N p\sbra{x_i;\tau,\epsilon,\eta}\\
&=\prod_{i=1}^N\int_0^{\infty} \cN\sbra{x_i|0,\alpha_i} \Gamma_{\tau}\sbra{\alpha_i|\epsilon,\eta}d\alpha_i.\end{split} \label{formu:marginaldist}}

In the following we study some special cases and show that the G-STG prior generalizes those in \cite{babacan2010bayesian,pedersen2011sparse}.

\emph{1) $\epsilon=1$: } In this case, the second layer is reduced to an exponential prior independent of $\tau$ since $\Gamma_{\tau}\sbra{\alpha_i|1,\eta}=\eta e^{-\eta \alpha_i}$. Then, the G-STG prior coincides with the Laplace prior in \cite{babacan2010bayesian} with $p\sbra{x_i;\tau,1,\eta}=\sqrt{\eta/2}\exp\sbra{-\sqrt{2\eta}\abs{x_i}}$.

\emph{2) $\tau=0$: } The second layer becomes a gamma prior for each $\alpha_i$. As a result, the proposed prior becomes the Gaussian-gamma prior in \cite{pedersen2011sparse} and $p\sbra{x_i}=\frac{2^{3/4-\epsilon/2}} {\sqrt{\pi}\Gamma\sbra{\epsilon}}\eta^{\frac{2\epsilon+1}{4}}\abs{x_i}^{\epsilon-\frac{1}{2}} \cK_{\epsilon-\frac{1}{2}}\sbra{\sqrt{2\eta}\abs{x_i}}$ where $\cK_{\nu}\sbra{\cdot}$ is the modified Bessel function of the second kind and order $\nu\in\bR$. In addition, we have that $p\sbra{0}=+\infty$ if $\epsilon\leq\frac{1}{2}$ and $p\sbra{0}<+\infty$ if $\epsilon>\frac{1}{2}$. Though the G-STG prior generalizes the Gaussian-gamma prior, it should be noted that the main algorithm based on the G-STG prior proposed in this paper works differently from that in \cite{pedersen2011sparse}.

\emph{3) $\tau\rightarrow+\infty$: } By l'Hospital's rule it can be shown that $\lim_{\tau\rightarrow+\infty}\Gamma_{\tau}\sbra{\alpha_i|\epsilon,\eta}=\eta e^{-\eta \alpha_i}$, i.e., the prior for $\alpha_i$ in the second layer approaches an exponential prior in such a case. Consequently, the proposed G-STG prior coincides with the Laplace prior as in \emph{Case 1}.

To visualize the variation of the G-STG prior with respect to the two parameters $\tau$ and $\epsilon$, we plot the PDF $p\sbra{x_i;\tau,\epsilon,\eta}$ in Fig. \ref{Fig:priorpdf} with $\eta=1$. Fig. \ref{Fig:priorpdf_tau} is for the case of $\epsilon=0.1$ and varying $\tau$. Obviously, the G-STG prior is a sparsity-inducing prior with the PDFs highly peaked at the origin, especially when $\tau\rightarrow0$. The main difference between the cases $\tau=0$ and $\tau=1\times10^{-8}$ is near the origin where $p\sbra{x_i}$ approaches infinity for $\tau=0$ while it is always finite for $\tau>0$. As $\tau$ gets larger, less density concentrates near the origin and the resulting prior gets closer to the Laplace prior that corresponds to $\tau=+\infty$. Fig. \ref{Fig:priorpdf_epsilon} is for the case of $\tau=1\times10^{-8}$ and varying $\epsilon$. It is shown that the G-STG prior gets less sparsity-inducing as $\epsilon$ gets larger, and that it ceases to promote sparsity as $\epsilon>1$. From the perspective of MAP estimation, the G-STG prior corresponds to a nonconvex optimization method as $\epsilon<1$ since the term $-\log p\sbra{\m{x}}$ is nonconvex in such a case.

\begin{figure}
\centering
  \subfigure[$\epsilon=0.1$]{
    \label{Fig:priorpdf_tau}
    \includegraphics[width=3.5in]{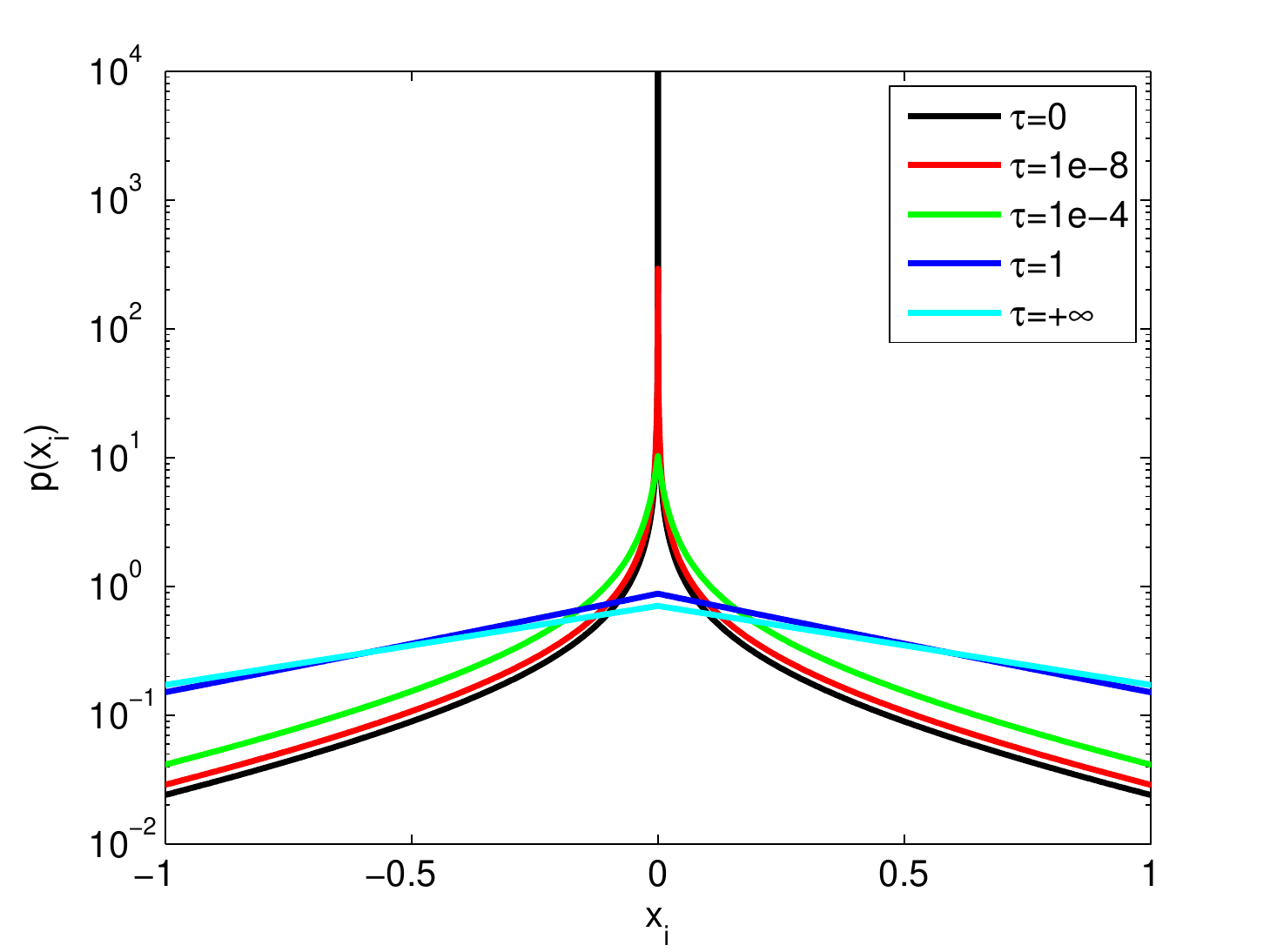}}
  \subfigure[$\tau=1\times10^{-8}$]{
    \label{Fig:priorpdf_epsilon}
    \includegraphics[width=3.5in]{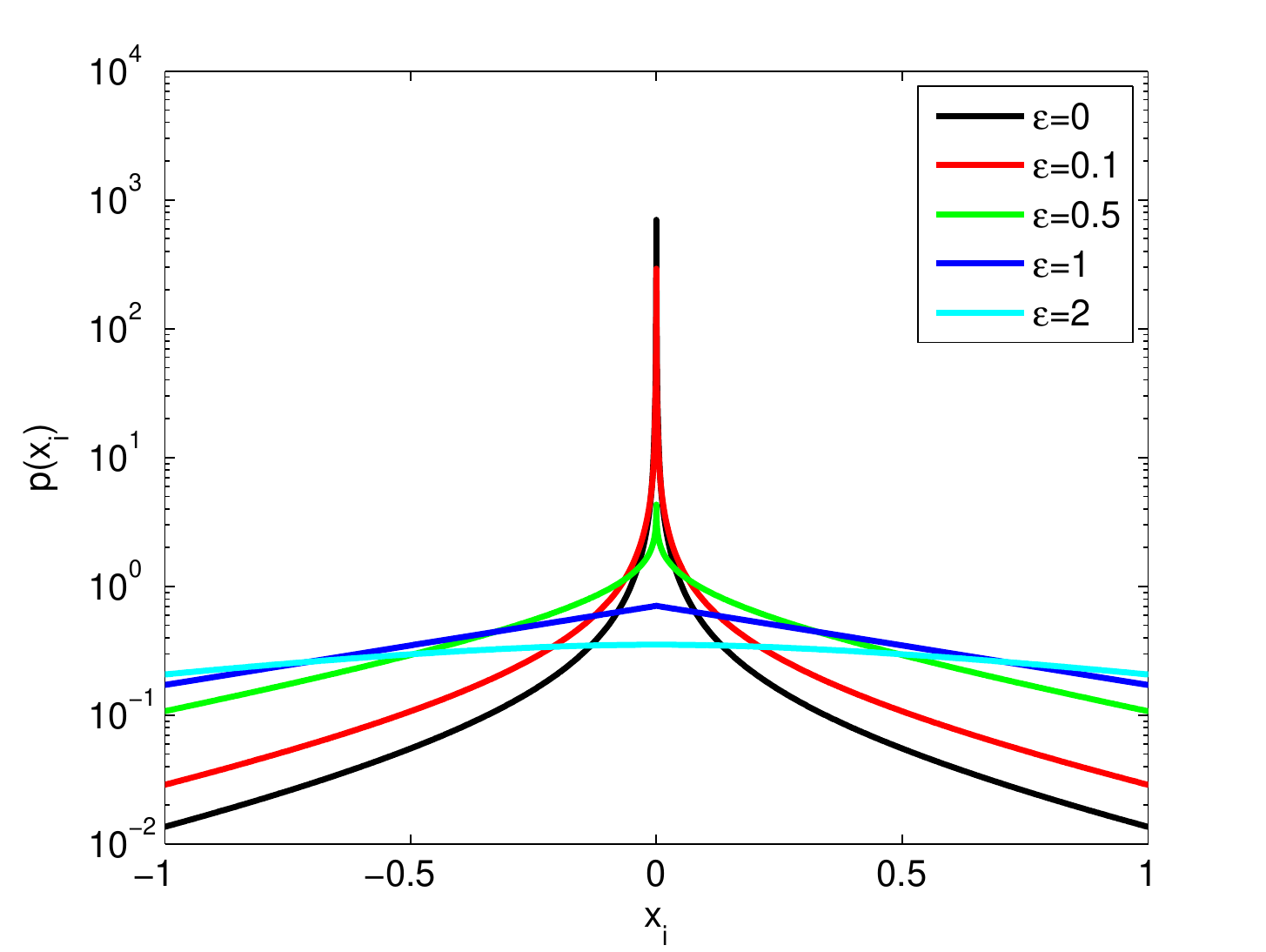}}
\centering
\caption{PDFs of the G-STG prior in the case of (a) $\epsilon=0.1$ and varying $\tau$ and (b) $\tau=1\times10^{-8}$ and varying $\epsilon$ with $\eta=1$.} \label{Fig:priorpdf}
\end{figure}

\subsection{Intuitive Interpretation and Threshold Parameter Setting} \label{sec:tau_setting}
Strictly speaking, a continuous prior is not suitable for sparse signals since any vector generated from a continuous prior is only approximately sparse (the probability of a zero-valued entry is zero). In the following, we provide an intuitive explanation about why the G-STG prior works for sparse signals by setting the threshold parameter $\tau$ according to the noise level.

We consider a Gaussian ensemble sensing matrix $\m{A}$ (the entries of $\m{A}$ are i.i.d. Gaussian $\cN\sbra{0,M^{-1}}$ where the variance is set to $M^{-1}$ to make columns of $\m{A}$ have expected unit norm). This matrix ensemble has been widely studied, e.g., in \cite{baraniuk2008simple,donoho2009message}. Then, consider a compressible signal $\m{z}$ and the observation model $\m{y}=\m{A}\m{z}$. In the Bayesian framework, we assume that $\m{z}$ is distributed according to some sparsity-inducing prior. Here we adopt the Gaussian-gamma prior in \cite{pedersen2011sparse}, i.e., we assume that $p\sbra{\m{z}|\m{\beta}}=\cN\sbra{\m{z}|\m{0},\m{\cB}}$ and $p\sbra{\m{\beta};\epsilon,\eta}=\prod_{i=1}^N\Gamma\sbra{\beta_i|\epsilon,\eta}$ where $\m{\cB}=\diag\sbra{\m{\beta}}$ and $\Gamma\sbra{\beta_i|\epsilon,\eta}$ refers to $\Gamma_{\tau}\sbra{\beta_i|\epsilon,\eta}$ with $\tau=0$. However, theoretical results \cite{candes2006stable,blumensath2009iterative} state that only significant entries of $\m{z}$ can be recovered while insignificant ones play as noises. So we write $\m{z}$ into $\m{z}=\m{x}+\m{w}$ where $\m{x}$ denotes the significant, ``recoverable'' component and $\m{w}$ refers to the insignificant, ``unrecoverable'' part. Then the observation model becomes $\m{y}=\m{A}\sbra{\m{x}+\m{w}}=\m{A}\m{x}+\m{e}$ with $\m{e}=\m{A}\m{w}$. By the structure of $\m{A}$, $\m{e}$ is a zero-mean AWGN with the noise variance $\sigma^2=M^{-1}\twon{\m{w}}^2$. We see that only the power of $\m{w}$ is reflected in the noise. That is, for any vector $\widetilde{\m{w}}$ satisfying $\twon{\widetilde{\m{w}}}^2=\twon{\m{w}}^2$, $\m{A}\widetilde{\m{w}}$ and $\m{e}$ are identically distributed. So we may replace $\m{w}$ by $\widetilde{\m{w}}$ in the observation model (i.e., $\m{x}+\widetilde{\m{w}}$ and $\m{z}$ are indistinguishable for CS approaches). Then we may model $\widetilde{\m{w}}$ as an i.i.d. zero-mean Gaussian vector with variance $\tau=N^{-1}\twon{\m{w}}^2=\sbra{M/N}\sigma^2$. In addition, $\widetilde{\m{w}}$ is independent of $\m{x}$. So, under the assumption of a Gaussian ensemble matrix $\m{A}$, a compressible signal $\m{z}$ is equivalent to the sum of its significant part $\m{x}$ plus a white Gaussian noise $\widetilde{\m{w}}$. Applying the Gaussian-gamma prior for $\m{z}$ to $\m{x}+\widetilde{\m{w}}$, we obtain that $\beta_i=\alpha_i+\tau$, $i=1,\cdots,N$, where $\m{\alpha}$ is as defined in (\ref{formu:prior_x}). Then we get $p\sbra{\alpha_i}=p\sbra{\beta_i-\tau|\beta_i\geq\tau}$ as a conditional distribution, resulting in that $p\sbra{\m{\alpha}}$ is in the exact form of (\ref{formu:prior_alpha}).

In this paper, we mainly consider the observation model $\m{y}=\m{A}\m{x}+\m{e}$ where $\m{x}$ is a sparse signal and $\m{e}$ is a zero-mean AWGN with known variance $\sigma^2$. The same sparsity-inducing prior for $\m{x}$ can be obtained by a reverse procedure and the details are omitted. So we can set the threshold parameter $\tau=\sbra{M/N}\sigma^2$ in the G-STG prior. Though this setting is only based on intuition without rigorous analysis, it indeed leads to good performance as to be reported via simulations in Section \ref{sec:simulation}, where it is also observed that this setting applies to other sensing matrix ensembles besides the Gaussian one.

\section{Sparse Bayesian Learning for Signal Recovery} \label{sec:algorithms}

\subsection{Bayesian Formulation}
In SBL, the signal of interest and noise are modeled as random variables. Under the common assumption of zero-mean AWGNs, i.e., $\m{e}\sim\cN\sbra{\m{0},\sigma^2\m{I}}$, where $\sigma^2$ is the noise variance, the PDF of the compressive measurements is
\equ{p\sbra{\m{y}|\m{x};\sigma^2}=\cN\sbra{\m{y}|\m{A}\m{x},\sigma^2\m{I}}.}
In this paper we assume that the noise variance $\sigma^2$ is known \emph{a priori}. This assumption has been widely made in the CS literature, e.g., \cite{candes2006stable,candes2011probabilistic}. Moreover, it is shown in \cite{wipf2007empirical} that to estimate $\sigma^2$ jointly with the signal recovery process (e.g., in \cite{tipping2001sparse}) can lead to very inaccurate estimate.

The G-STG prior introduced in Section \ref{sec:newsparseprior} is adopted as the sparsity-inducing prior for the sparse signal $\m{x}$. The hyperparameters $\tau$ and $\epsilon$ are chosen manually according to the reasoning in Section \ref{sec:newsparseprior} and Subsection \ref{sec:basisselectcondition}. Numerical simulations will be provided in Section \ref{sec:simulation} to illustrate their performance. To estimate $\eta$ from the measurements, we assume a gamma hyperprior for $\eta$: $p\sbra{\eta;c,d}=\Gamma\sbra{\eta|c,d}$, where we let $c,d\rightarrow0$ to obtain a uniform hyperprior (over a logarithmic scale). So the joint PDF of the observation model is
$p\sbra{\m{y},\m{x},\m{\alpha},\eta;\sigma^2,\tau,\epsilon,c,d}=p\sbra{\m{y}|\m{x};\sigma^2} p\sbra{\m{x}|\m{\alpha}}p\sbra{\m{\alpha}|\eta;\tau,\epsilon}p\sbra{\eta;c,d}$, where $\m{y}$ is the observation, $\m{x}$ is the unknown signal of interest, $\m{\alpha}$ and $\eta$ are unknown parameters, and $\sigma^2$, $\tau$, $\epsilon$, $c$ and $d$ are fixed. The task is to estimate $\m{x}$.

\subsection{Bayesian Inference} \label{sec:bayesianinference}
Note that the exact Bayesian inference is intractable since $p\sbra{\m{x}|\m{y}}$ is computationally intractable. Some approximations have to be made. Following from \cite{tipping2001sparse}, we decompose the posterior $p\sbra{\m{x},\m{\alpha},\eta|\m{y}}$ into two terms as
\equ{p\sbra{\m{x},\m{\alpha},\eta|\m{y}} = p\sbra{\m{x}|\m{y},\m{\alpha}} p\sbra{\m{\alpha},\eta|\m{y}}.}
The first term $p\sbra{\m{x}|\m{y},\m{\alpha}}$ is the posterior for $\m{x}$ given the hyperparameter $\m{\alpha}$, which will be later shown to be a Gaussian PDF. Then, we compute the most probable estimates of $\m{\alpha}$ and $\eta$, say $\m{\alpha}_{MP}$ and $\eta_{MP}$, that maximize the second term $p\sbra{\m{\alpha},\eta|\m{y}}$. We use $\m{\alpha}_{MP}$ to obtain the posterior for $\m{x}$. From the perspective of signal estimation, this is equivalent to requiring
\equ{p\sbra{\m{x}|\m{y}}=\int p\sbra{\m{x},\m{\alpha},\eta|\m{y}}d\m{\alpha}\,d\eta \approx p\sbra{\m{x}|\m{y},\m{\alpha}_{MP}},}
where $p\sbra{\m{x}|\m{y},\m{\alpha}_{MP}}$ refers to $p\sbra{\m{x}|\m{y},\m{\alpha}}$ at $\m{\alpha}_{MP}$. Similar procedures have been adopted in \cite{wipf2004sparse,babacan2010bayesian}. We will provide theoretical evidence to show that this approach leads to desirable properties in Section \ref{sec:analysis}. Simulation results presented in Section \ref{sec:simulation} also suggest that the signal recovery based on this approximation is very effective.

Since $p\sbra{\m{y}|\m{x}}$ and $p\sbra{\m{x}|\m{\alpha}}$ are both Gaussian, it is easy to show that the posterior for $\m{x}$ and the marginal distribution for $\m{y}$ are both Gaussian with
$p\sbra{\m{x}|\m{y},\m{\alpha}}= \cN\sbra{\m{x}|\m{\mu},\m{\Sigma}}$,  $p\sbra{\m{y}|\m{\alpha}}= \cN\sbra{\m{y}|\m{0},\m{C}}$, where
{\lentwo\equa{\m{\mu}
&=& \sigma^{-2}\m{\Sigma}\m{A}^T\m{y}, \label{formu:update_mu}\\ \m{\Sigma}
&=& \sbra{\sigma^{-2}\m{A}^T\m{A}+\m{\Lambda}^{-1}}^{-1}, \label{formu:update_Sigma}\\ \m{C}
&=& \sigma^2\m{I}+\m{A}\m{\Lambda}\m{A}^T.\label{formu:update_C}}
}The maximization of $p\sbra{\m{\alpha},\eta|\m{y}}$ is equivalent to that of $p\sbra{\m{y},\m{\alpha},\eta}=p\sbra{\m{y}|\m{\alpha}} p\sbra{\m{\alpha}|\eta}p\sbra{\eta}$ by the relation $p\sbra{\m{\alpha},\eta|\m{y}}=p\sbra{\m{y},\m{\alpha},\eta}/p\sbra{\m{y}}$. We consider $\log\eta$ as the hidden variable instead of $\eta$ since the uniform hyperprior is assumed over a logarithmic scale. By $p\sbra{\log\eta}=\eta p\sbra{\eta}$ we see that the hyperprior for $\eta$ leads to a noninformative prior by setting $c=d=0$. So the log-likelihood function is
\equ{\begin{split}
&\cL\sbra{\m{\alpha},\log\eta}\\
&= \log p\sbra{\m{y},\m{\alpha},\log\eta} \\
&= -\frac{1}{2}\log\abs{\m{C}}-\frac{1}{2}\m{y}^T\m{C}^{-1}\m{y} \\ &\quad+\sbra{\epsilon-1}\sum_{i=1}^N\log\sbra{\alpha_i+\tau}-\eta\sum_{i=1}^N\sbra{\alpha_i+\tau}\\
&\quad+ \sbra{N\epsilon+c}\log\eta-N\log\Gamma_{\eta\tau}\sbra{\epsilon} -d\eta + C_1, \end{split} \label{formu:likelihood}}
where $C_1$ is a constant. The maximizer of $\cL$ will be analyzed in Subsection \ref{sec:analysis}. In the following we provide an iterative procedure to maximize $\cL$ by recognizing the identities $\log\abs{\m{C}}=\log\abs{\m{\Lambda}} + M\log\sigma^2-\log\abs{\m{\Sigma}}$ and $\m{y}^T\m{C}^{-1}\m{y}= \sigma^{-2}\twon{\m{y}-\m{A}\m{\mu}}^2 + \m{\mu}^T\m{\Lambda}^{-1}\m{\mu}$.

\subsubsection{Update of $\m{\alpha}$}For $\alpha_i$, $i=1,\cdots,N$, we have
\equ{\begin{split}\frac{\partial\cL}{\partial\alpha_i}
&=-\frac{1}{2\alpha_i}+\frac{E\lbra{x_i^2}}{2\alpha_i^2} +\frac{\epsilon-1}{\alpha_i+\tau}-\eta\\
&=-\frac{f\sbra{\alpha_i}}{2\alpha_i^2\sbra{\alpha_i+\tau}},\end{split}}
where $f\sbra{t}=2\eta t^3+\sbra{3-2\epsilon+2\eta\tau}t^2 +\sbra{\tau-E\lbra{x_i^2}}t-\tau E\lbra{x_i^2}$ for $t\in\bR$ is a cubic function and $E\lbra{x_i^2}=\mu_i^2+\Sigma_{ii}$, where $\Sigma_{ii}$ is the $i$th diagonal entry of $\m{\Sigma}$. We need the following lemma.

\begin{lem} For a cubic function $g\sbra{t}=\lambda_1t^3+\lambda_2t^2+\lambda_3t+\lambda_4$, if $\lambda_1, \lambda_2>0$ and $\lambda_4<0$, then $g\sbra{t}=0$ has a unique root on $\left(0,+\infty\right)$. \label{lem:rubicroot}
\end{lem}
\begin{proof} By $g(0)=\lambda_4<0$ and $\lim_{t\rightarrow+\infty}g\sbra{t}=+\infty$ there exists at least one root in $\left(0,+\infty\right)$. We show that this root is unique using contradiction. Suppose there exists more than one positive root. Then there must exist three positive roots and that $g\sbra{t}$ has two positive stationary points. That is, the two solutions of $\frac{dg(t)}{dt}=3\lambda_1t^2+2\lambda_2t+\lambda_3=0$ are both positive, resulting in that $\lambda_2<0$ (contradiction).
\end{proof}

By Lemma \ref{lem:rubicroot}, it is easy to show that the maximum of $\cL$ is achieved at the unique positive root, say $\alpha_i^*>0$, of $f\sbra{\alpha_i}=0$. We note that explicit expressions are available for the roots of a cubic function and hence $\alpha_i$, $i=1,\cdots,N$, can be efficiently updated.

\subsubsection{Update of $\eta$} In general, there is no explicit expression for updating $\eta$. Since the first and second derivatives of $\cL$ with respect to $\log\eta$ can be easily computed, $\cL$ can be efficiently maximized with respect to $\log\eta$ using numerical methods, e.g., gradient ascending method or Newton's method. In addition, the computational complexity hardly depends on the CS problem dimension.

\begin{rem} The computation of the incomplete gamma function $\Gamma_{\eta\tau}\sbra{\epsilon}$ is involved in the update of $\eta$. This term can be efficiently computed using functions provided in Matlab if $\epsilon$ is properly bounded away from zero. But a numerical integration is needed if $\epsilon=0$. In Section \ref{sec:simulation}, we observe through simulations that the update of $\eta$ may take considerably long time in the case of $\epsilon=0$. But such differences are negligible in the case of a high-dimensional CS problem since the computation of $\eta$ hardly depends on the problem dimension unlike other computations, such as the update of $\m{\alpha}$. \label{rem:timeatepsilonis0}
\end{rem}



As a result, an iterative algorithm can be implemented to obtain $\m{\alpha}_{MP}$ and $\eta_{MP}$ by iteratively updating $\m{\Sigma}$ in (\ref{formu:update_Sigma}), $\m{\mu}$ in (\ref{formu:update_mu}), $\m{\alpha}$ and $\eta$. It is easy to show that this algorithm can be implemented using an EM algorithm \cite{mclachlan2008algorithm}. So the likelihood $\cL$ increases monotonically at each iteration and the algorithm is guaranteed to converge. After convergence, the signal $\m{x}$ is estimated using its posterior mean $\m{\mu}$. One shortcoming of the iterative algorithm is that at each iteration a high-dimensional matrix inversion has to be calculated for updating $\m{\Sigma}$ though this computation can be possibly alleviated using the Woodbury matrix identity.

\subsection{Analysis of Global and Local Maxima} \label{sec:analysis}
We analyze the global and local maxima of the likelihood $\cL$ in (\ref{formu:likelihood}) in this subsection. Our analysis is rooted in \cite{wipf2004sparse} and shows that the theoretical results on the basic SBL in \cite{wipf2004sparse} can be extended to our case with necessary modifications. In the following, we assume that $c=d=0$ and $\eta>0$ is fixed (it is a similar case if $\eta$ is chosen to maximize $\cL$ as well). Thus we may write $\cL$ (with respect to $\m{\alpha}$) as
\equ{\begin{split}\cL\sbra{\m{\alpha}}
&= -\frac{1}{2}\log\abs{\m{C}}-\frac{1}{2}\m{y}^T\m{C}^{-1}\m{y} \\ &\quad+\sbra{\epsilon-1}\sum_{i=1}^N\log\sbra{\alpha_i+\tau}-\eta\sum_{i=1}^N\alpha_i + C_4, \end{split} \label{formu:likelihood2}}
where $C_4$ is a constant independent of $\m{\alpha}$.
We first consider the global maxima in the noise free case.

\begin{thm} Let $\tau\geq0$, $0\leq\epsilon\leq1$, $\m{e}=0$ and $\sigma^2=0$. Assume that $\m{x}^0$ satisfying $\zeron{\m{x}^0}<M<N$ is the maximally sparse solution to the linear equation $\m{y}=\m{A}\m{x}$. Then, there exists some $\m{\alpha}^0$ with $\twon{\m{\alpha}^0}<+\infty$ and $\zeron{\m{\alpha}^0}=\zeron{\m{x}^0}$ such that at $\m{\alpha}^0$, $\cL$ is globally maximized and the corresponding posterior mean is $\m{\mu}=\m{x}^0$. \label{thm:globalmaxima}
\end{thm}

\begin{proof} Note first that $\lim_{\sigma^2\rightarrow0}\m{\mu}=\m{\Lambda}^{\frac{1}{2}} \sbra{\m{A}\m{\Lambda}^{\frac{1}{2}}}^\dagger\m{y}=\m{\Lambda}^{\frac{1}{2}} \sbra{\m{A}\m{\Lambda}^{\frac{1}{2}}}^\dagger\m{A}\m{x}^0$ by (\ref{formu:update_mu}). Since $\zeron{\m{x}^0}<M$ there must exist $\m{\alpha}^0\in\bR^N$, $\twon{\m{\alpha}^0}<+\infty$ and $\zeron{\m{\alpha}^0}=\zeron{\m{x}^0}$, such that $\m{\mu}=\m{x}^0$ at $\m{\alpha}=\m{\alpha}^0$. In fact, $\m{\alpha}^0$ can be any vector satisfying that for $j=1,\cdots,N$, $\alpha_j^0>0$ if $x_j^0\neq0$, and $\alpha_j^0=0$ otherwise. In the following we show that the global maximum of $\cL$ is achieved at $\m{\alpha}=\m{\alpha}^0$. Let $C_5=\inftyn{\m{\alpha}^0}<\infty$ and denote $C_6>0$ the minimum nonzero $\alpha_i$, $i=1,\cdots,N$. Following from the proof of \cite[Theorem 1]{wipf2004sparse}, we have $\abs{\m{C}}=\abs{\sigma^2\m{I}+\m{A}\m{\Lambda}\m{A}^T}\rightarrow0$ and $\m{y}^T\m{C}^{-1}\m{y}\leq\twon{\m{x}^0}^2/C_6$ if $\m{\alpha}\rightarrow\m{\alpha}^0$. In addition, we have $\sum_{i=1}^N\log\sbra{\alpha_i+\tau}\leq N\log\sbra{C_5+\tau}$ and $\sum_{i=1}^N\alpha_i\leq NC_5$. So we have $\cL=+\infty$ at $\m{\alpha}=\m{\alpha}^0$, which completes the proof.
\end{proof}

Theorem \ref{thm:globalmaxima} shows that the global maximum of the objective function is achieved at the maximally sparse solution. Thus the proposed approach has no structural errors and the remaining question of whether the algorithm can produce this solution is a convergence issue. Further, we have the following result which is similar to that in \cite[Theorem 2]{wipf2004sparse}.

\begin{thm} Let $\tau\geq0$ and $0\leq\epsilon\leq1$. Every local maximum of $\cL$ is achieved at a sparse $\m{\alpha}$ with $\zeron{\m{\alpha}}\leq M$ that leads to a sparse posterior mean $\m{\mu}$ with $\zeron{\m{\mu}}\leq M$, regardless of the existence of noise. \label{thm:localmaxima}
\end{thm}

\begin{proof} Note that $\zeron{\m{\mu}}\leq\zeron{\m{\alpha}}$ since $\mu_i\rightarrow0$ as $\alpha_i\rightarrow0$, $i=1,\cdots,N$. So we need only to show that every local maximum of $\cL$ is achieved at a sparse $\m{\alpha}$ with $\zeron{\m{\alpha}}\leq M$. Let $q\sbra{\m{\alpha}}=-\frac{1}{2}\log\abs{\m{C}}+\sbra{\epsilon-1}\sum_{i=1}^N \log\sbra{\alpha_i+\tau}-\eta\sum_{i=1}^N\alpha_i$, which is convex with respect to $\m{\alpha}\in\bR^N_+$ if $\epsilon\leq1$. Suppose that $\m{\alpha}^*$ is a local maximum point of $\cL$. We may construct a closed, bounded convex polytope $\cP\subset\bR^N_+$ following from \cite{wipf2004sparse} (we omit the details) such that $\m{\alpha}^*\in\cP$ and if $\m{\alpha}\in\cP$, then the second term of $\cL$, $-\frac{1}{2}\m{y}^T\m{C}^{-1}\m{y}$, equals a constant $C_7=-\frac{1}{2}\m{y}^T\sbra{\sigma^2\m{I}+\m{A}\diag\sbra{\m{\alpha}^*}\m{A}^T}^{-1}\m{y}$. In addition, all extreme points of $\cP$ are sparse with support size no more than $M$. As a result, $\m{\alpha}^*$ is a local maximum point of $q\sbra{\m{\alpha}}$ with respect to $\m{\alpha}\in\cP$. By the convexity of $q\sbra{\m{\alpha}}$, $\m{\alpha}^*$ must be an extreme point of $\cP$ with $\zeron{\m{\alpha}^*}\leq M$.
\end{proof}

Theorem \ref{thm:localmaxima} states that all local maxima of $\cL$ are achieved at sparse solutions. Since we can locally maximize $\cL$ efficiently in practice, based on Theorem \ref{thm:localmaxima}, we introduce a fast algorithm in Section \ref{sec:greedyalg} that searches for a sparse solution that locally maximizes $\cL$.

\section{Fast Algorithm} \label{sec:greedyalg}

\subsection{Fast Greedy Algorithm} \label{sec:greedyalgorithm}
Based on Theorem \ref{thm:localmaxima} in Section \ref{sec:analysis}, the following algorithm aims to find a sparse solution that locally maximizes $\cL$. We consider the contribution of a single basis vector $\m{a}_j$ (the $j$th column of $\m{A}$), $j=1,\cdots,N$, and determine whether it should be included in the model (or in the active set) for the maximization of $\cL$. Denote $\m{C}_{-j}=\sigma^2\m{I}+\sum_{i\neq j}\alpha_i\m{a}_i\m{a}_i^T=\m{C}-\alpha_j\m{a}_j\m{a}_j^T$ that is independent of the $j$th basis vector $\m{a}_j$. Then we have $\abs{\m{C}}=\abs{\m{C}_{-j}} \abs{1+\alpha_j\m{a}_j^T\m{C}_{-j}^{-1}\m{a}_j}$ and $\m{C}^{-1}=\m{C}_{-j}^{-1}- \sbra{\alpha_j^{-1}+\m{a}_j^T\m{C}_{-j}^{-1}\m{a}_j}^{-1} \m{C}_{-j}^{-1}\m{a}_j\m{a}_j^T\m{C}_{-j}^{-1}$. Using the two identities above, we rewrite the log-likelihood function (with respect to $\m{\alpha}$) into
\equ{\begin{split}
\cL\sbra{\m{\alpha}}
&= -\frac{1}{2}\log\abs{\m{C}_{-j}}-\frac{1}{2}\m{y}^T\m{C}_{-j}^{-1}\m{y} \\ &\quad+\sbra{\epsilon-1}\sum_{i=1}^N\log\sbra{\alpha_i+\tau}-\eta\sum_{i=1}^N\alpha_i \\
&\quad-\frac{1}{2}\log\abs{1+\alpha_j\m{a}_j^T\m{C}_{-j}^{-1}\m{a}_j} +\frac{q_j^2}{2\sbra{\alpha_j^{-1}+s_j}}+C_2\\
&=\cL\sbra{\m{\alpha}_{-j}}+\ell\sbra{\alpha_j}+C_3,\end{split}}
where $\cL\sbra{\m{\alpha}_{-j}}$ denotes $\cL\sbra{\m{\alpha}}$ after removing the contribution of the $j$th basis vector, $\ell\sbra{\alpha_j}=-\frac{1}{2}\log\abs{1+\alpha_js_j} +\frac{q_j^2}{2\sbra{\alpha_j^{-1}+s_j}} +\sbra{\epsilon-1}\log\sbra{\frac{\alpha_j}{\tau}+1}-\eta\alpha_j$ with $s_j=\m{a}_j^T\m{C}_{-j}^{-1}\m{a}_j$ and $q_j=\m{a}_j^T\m{C}_{-j}^{-1}\m{y}$, and $C_2$, $C_3$ are constants independent of $\m{\alpha}$. Note that $\ell\sbra{\alpha_i}$ has been modified by a constant such that $\ell\sbra{0}=0$. In the following we compute the maximum point, say $\alpha_j^*$, of $\ell\sbra{\alpha_j}$ on $\left[0,+\infty\right)$. If $\alpha_j^*>0$, then the basis vector $\m{a}_j$ should be preserved in the active set since it gives positive contribution to the likelihood. Otherwise, it should be removed from the active set (by setting $\alpha_j=0$). We consider only the general case $\tau>0$, $0\leq\epsilon\leq1$ and $\eta>0$ since it is simple for other cases. First we have
\equ{\begin{split}\frac{d\,\ell\sbra{\alpha_j}}{d\,\alpha_j}
&=-\frac{s_j}{2\sbra{1+\alpha_js_j}} +\frac{q_j^2}{2\sbra{1+\alpha_js_j}^2} +\frac{\epsilon-1}{\alpha_j+\tau}-\eta\\
&= -\frac{h\sbra{\alpha_j}}{2\sbra{\alpha_j+\tau}\sbra{1+\alpha_js_j}^2},\end{split}}
where $h\sbra{\alpha_j}=c_1\alpha_j^3+c_2\alpha_j^2+c_3\alpha_j+c_4$ is a cubic function of $\alpha_j$ with
{\lentwo\equa{c_1
&=& 2\eta s_j^2,\\ c_2
&=& \sbra{3-2\epsilon}s_j^2+4\eta s_j+2\eta\tau s_j^2,\\ c_3
&=& \sbra{5-4\epsilon}s_j+2\eta-q_j^2+\tau\sbra{4\eta s_j+s_j^2},\\ c_4
&=& 2-2\epsilon+\tau\sbra{s_j+2\eta-q_j^2}.}
}To compute the maximum point of $\ell\sbra{\alpha_j}$, we need the following result.
\begin{lem}[\cite{irving2004integers}] For a cubic function $g\sbra{t}=\lambda_1t^3+\lambda_2t^2+\lambda_3t+\lambda_4$, let the discriminant $\Delta=18\lambda_1\lambda_2\lambda_3\lambda_4 -4\lambda_2^3\lambda_4 +\lambda_2^2\lambda_3^2 -4\lambda_1\lambda_3^3 -27\lambda_1^2\lambda_4^2$. If $\Delta>0$, then $g\sbra{t}=0$ has three distinct real roots. If $\Delta=0$, then the equation has three real roots including a multiple root. If $\Delta<0$, then the equation has one real root and two complex conjugate roots. \label{lem:discriminant}
\end{lem}

Note that $s_j>0$ for $j=1,\cdots,N$ since $\m{C}_{-j}$ and $\m{C}_{-j}^{-1}$ are positive definite,  and then $c_1,c_2>0$. We divide our discussions into three scenarios.
\subsubsection{$c_4<0$} \label{sec:basisanalysis1}
This case is the same as the update of $\m{\alpha}$ in Subsection \ref{sec:bayesianinference}. By Lemma \ref{lem:rubicroot}, $h\sbra{\alpha_j}=0$ has a unique solution $\alpha_j^*$ on $\left(0,+\infty\right)$. Then it is easy to show that $\ell\sbra{\alpha_j}$ increases monotonically on $\left(0,\alpha_j^*\right]$ and decreases monotonically on $\sbra{\alpha_j^*,+\infty}$. Thus $\ell\sbra{\alpha_j}$ obtains the maximum at $\alpha_j^*>0$.

\subsubsection{$c_4\geq0$, $c_3<0$ and $\Delta>0$} \label{sec:basisanalysis2}
Let $\Delta$ denote the determinant of $h\sbra{\alpha_j}$. We see that $h\sbra{\alpha_j}=0$ has three real distinct roots by Lemma \ref{lem:discriminant} since $\Delta>0$. The two stationary points of $h\sbra{\alpha_j}$ lie on different sides of the $y$-axis since $\frac{d\,h\sbra{\alpha_j}}{d\,\alpha_j}=3c_1\alpha_j^2+2c_2\alpha_j+c_3$ with $c_3<0$. Thus the three roots include one negative root and two distinct positive roots since $h(0)=c_4\geq0$. Denote $\alpha'_j>0$ the largest root. We see that $\ell\sbra{\alpha_j}$ decreases from $\alpha_j=0$ to some point, then increases until $\alpha_j=\alpha'_j$ and decreases again. As a result, $\ell\sbra{\alpha_j}$ obtains the maximum at $0$ or $\alpha'_j$. So we have
\begin{itemize}
 \item if $\ell\sbra{\alpha'_j}>0$, then the maximum point is $\alpha_j^*=\alpha'_j>0$;
 \item if $\ell\sbra{\alpha'_j}\leq0$, then $\alpha_j^*=0$;
\end{itemize}

\subsubsection{$c_4\geq0$, and $c_3\geq0$ or $\Delta\leq0$} \label{sec:basisanalysis3}
In this scenario, the maximum of $\ell\sbra{\alpha_j}$ is obtained at $\alpha_j^*=0$. We divide our discussions into three cases: \emph{1)} $\Delta<0$, \emph{2)} $\Delta=0$, and \emph{3)} $\Delta>0$ and $c_3\geq0$. In \emph{Case 1}, $h\sbra{\alpha_j}=0$ has only one (negative) real root. In \emph{Case 2}, the equation has three negative roots (two of them coincide), or one negative root and a multiple positive root. In \emph{Case 3} the equation has three distinct negative roots. So in all the cases $\ell'\sbra{\alpha_j}\leq0$ for $\alpha_j\geq0$, resulting in that $\ell\sbra{\alpha_j}$ decreases monotonically on $\left[0,+\infty\right)$.

Based on the analysis above, we can compute efficiently $\alpha_j^*$ (given $s_j$ and $q_j$) at which the likelihood is maximized with respect to a single basis vector $\m{a}_j$, $j=1,\cdots,N$. So, the likelihood consistently increases if we update a single $\alpha_j$ at one time with the basis $\m{a}_j$, $j=1,\cdots,N$, properly chosen. As a result, a greedy algorithm can be implemented. At the beginning, no basis vectors are included in the model (i.e., the active set is empty or all $\alpha_j=0$). Then choose a vector $\m{a}_j$ at each iteration such that it gives the largest likelihood increment by updating $\alpha_j$ from its current value $\alpha_j^0$ to the maximum point $\alpha_j^*$. If $\alpha_j^0=0$ and $\alpha_j^*>0$, then the basis vector $\m{a}_j$ is added to the model with $\alpha_j=\alpha_j^*$. If $\alpha_j^0>0$ and $\alpha_j^*>0$, then $\alpha_j$ is re-estimated in the model. If $\alpha_j^0>0$ and $\alpha_j^*=0$, then $\m{a}_j$ is removed from the model with $\alpha_j=0$. After that, update $\eta$ as in Subsection \ref{sec:bayesianinference}. The process is repeated until convergence (that is guaranteed since $\cL$ increases monotonically). Based on the results of \cite{tipping2003fast}, we see that $\m{\mu}$, $\m{\Sigma}$, $s_j$ and $q_j$, $j=1,\cdots,N$, can be efficiently updated without matrix inversions. Consequently, the greedy algorithm is computationally efficient.

\begin{rem} The main differences between the proposed algorithm and those in \cite{tipping2003fast,babacan2010bayesian} are the updates of $\m{\alpha}$ and $\eta$. Since roots of a cubic equation have explicit expressions and the update of $\eta$ hardly depends on the problem dimension, the proposed greedy algorithm has the same computational complexity as those in \cite{tipping2003fast,babacan2010bayesian} at each iteration.
\end{rem}

\subsection{Analysis of Basis Selection Condition} \label{sec:basisselectcondition}
We study the basis selection condition of the fast algorithm in more details in this subsection. Based on the analysis in Subsection \ref{sec:greedyalgorithm}, we have the following result.
\begin{prop} Suppose that the log-likelihood $\cL$ has been locally maximized in the fast algorithm. If $\alpha_j>0$ for some basis vector $\m{a}_j$, $j=1,\cdots,N$, then $q_j^2>\min\lbra{s_j+2\eta+\frac{2-2\epsilon}{\tau},\, \sbra{5-4\epsilon}s_j+2\eta+\tau\sbra{4\eta s_j+s_j^2}}$ where $s_j$, $q_j$ are as defined in Subsection \ref{sec:greedyalgorithm}. \label{prop:basisselcondition}
\end{prop}
\begin{proof} Note that the inequalities $q_j^2>s_j+2\eta+\frac{2-2\epsilon}{\tau}$ and $q_j^2>\sbra{5-4\epsilon}s_j+2\eta+\tau\sbra{4\eta s_j+s_j^2}$ are equivalent to $c_4<0$ and $c_3<0$ respectively. Let us suppose that the conclusion does not hold. Then we have $c_4\geq0$ and $c_3\geq0$. It follows from the analysis in Subsection \ref{sec:basisanalysis3} that the likelihood increases if $\m{a}_j$ is removed from the model, leading to contradiction.
\end{proof}

\begin{rem} Proposition \ref{prop:basisselcondition} provides a necessary condition for the basis vectors in the final active set. Note that this condition is generally insufficient since additional requirements are needed in the case of $c_4\geq0$ (e.g., $\Delta>0$) according to the analysis in Subsection \ref{sec:greedyalgorithm}. In a special case where $q_j^2>s_j+2\eta+\frac{2-2\epsilon}{\tau}$ holds, we can conclude that $\m{a}_j$ is in the model according to Subsection \ref{sec:basisanalysis1}. \label{rem:necessarynotsuff}
\end{rem}

The basis selection condition concerns the sparsity level of the solution of the fast algorithm. We have illustrated that different settings of $\tau$ and $\epsilon$ lead to different sparsity-inducing priors in Section \ref{sec:newsparseprior}. In the following we will see how the parameters affect the sparsity of the algorithm solution. We discuss the effects of $\tau$ and $\epsilon$ separately.

Let us first fix $\tau>0$ and vary $\epsilon\in\mbra{0,1}$. It is easy to show that as $\epsilon$ decreases both the terms $s_j+2\eta+\frac{2-2\epsilon}{\tau}$ and $\sbra{5-4\epsilon}s_j+2\eta+\tau\sbra{4\eta s_j+s_j^2}$ increase. Thus the necessary condition in Proposition \ref{prop:basisselcondition} becomes stronger. As a result, the solution of the greedy algorithm will be sparser, which is consistent with the fact that a smaller $\epsilon$ leads to a more sparsity-inducing prior as shown in Section \ref{sec:newsparseprior}. In the case of $\epsilon=1$, the inequality turns to be $q_j^2>s_j+2\eta$ that coincides with the result in \cite{babacan2010bayesian}. So, the proposed algorithm will produce a sparser solution than that of \cite{babacan2010bayesian} by simply setting $\epsilon<1$.

It is not obvious for the case of fixed $\epsilon<1$ and varying $\tau$ (the case $\epsilon=1$ has been discussed before) since, as $\tau$ decreases, the first term $s_j+2\eta+\frac{2-2\epsilon}{\tau}$ increases while the second term $\sbra{5-4\epsilon}s_j+2\eta+\tau\sbra{4\eta s_j+s_j^2}$ decreases. To make a correct conclusion, we observe that $\ell\sbra{\alpha_j}$ in Subsection \ref{sec:greedyalgorithm} is a strictly increasing function with respect to $\tau>0$ for any fixed $\alpha_j>0$ and keeps equal to zero at $\alpha_j=0$. As a result, as $\tau$ decreases, it is less likely that the maximum of $\ell\sbra{\alpha_j}$ is achived at a positive point, resulting in that the solution of the greedy algorithm gets sparser. This is consistent with that a smaller $\tau$ leads to a more sparsity-inducing prior as shown in Section \ref{sec:newsparseprior}. In fact, the greedy algorithm produces a zero solution if $\tau=0$ and $\epsilon<1$ since in such a case the maximum of the likelihood is always achieved at the origin with respect to every basis vector. So intuitively, a smaller $\tau$ should be used to obtain a more sparsity-inducing prior but too small $\tau$ may lead to inaccuracy for the greedy algorithm. Numerical simulations in Section \ref{sec:simulation} will illustrate that the recommended $\tau$ in Subsection \ref{sec:tau_setting} is a good choice.

\begin{rem} In the case of $\tau=0$, the proposed G-STG prior coincides with the Gaussian-gamma prior in \cite{pedersen2011sparse}. The greedy algorithm using this prior developed in \cite{pedersen2011sparse} is claimed to follow from the same framework in \cite{tipping2003fast} and maximize the likelihood sequentially. However, it should be noted that the algorithm in \cite{pedersen2011sparse} does not really maximize the likelihood since, if it does, then it should produce a zero solution as discussed above. Specifically, the authors of \cite{pedersen2011sparse} compute only a local maximum point of $\ell\sbra{\alpha_j}$ which cannot guarantee to increase the likelihood $\cL$ while the global maximum of $\ell\sbra{\alpha_j}$ is always obtained at the origin. Hence the algorithm in \cite{pedersen2011sparse} is technically incorrect because of the inappropriate basis update scheme. Moreover, the algorithm in \cite{pedersen2011sparse} has not been shown to provide guaranteed convergence.
\end{rem}

\section{Numerical Simulations} \label{sec:simulation}
In this section, we present numerical results to illustrate the performance of the proposed method (we consider only the fast algorithm in Subsection \ref{sec:greedyalgorithm}). We consider both one-dimensional synthetic signals and two-dimensional images, and compare with existing methods, including $\ell_1$ optimization (BP or BPDN), reweighted (RW-) $\ell_1$ optimization \cite{candes2008enhancing}, StOMP\cite{donoho2006sparse}, the basic BCS \cite{tipping2003fast} and BCS with the Laplace prior (denoted by Laplace)\cite{babacan2010bayesian}. BCS, Laplace and the proposed method are SBL methods. $\ell_1$ optimization is a convex optimization method. RW-$\ell_1$ is related to nonconvex optimization. StOMP is a greedy method. The Matlab codes of BP, StOMP and BCS are obtained from the SparseLab package\footnote{Available at http://sparselab.stanford.edu.}, and that of BPDN is from the $\ell_1$-magic package\footnote{Available at http://users.ece.gatech.edu/$\sim$justin/l1magic.}. The code of Laplace is available at https://netfiles.uiuc.edu/dbabacan/www/links.html. The number of iterations is set to 5 for RW-$\ell_1$ (i.e., 5 $\ell_1$ minimization problems are solved iteratively). To make a fair comparison, we use the same convergence criterion in the proposed method as in BCS and Laplace with the stopping tolerance set to $1\times10^{-8}$.

The performance metrics adopted include the relative mean squared error (RMSE, calculated by $\twon{\widehat{\m{x}}-\m{x}}^2/\twon{\m{x}}^2$), the support size ($\zeron{\widehat{\m{x}}}$), the number of iterations (for the three BCS methods) and the CPU time, where $\widehat{\m{x}}$ and $\m{x}$ denote the recovered and original signals, respectively.

\subsection{One-Dimensional Synthetic Signals}

\subsubsection{Performance with respect to $\tau$} An explicit determination of $\tau$ has been recommended in Subsection \ref{sec:tau_setting}. In the following, we show that, indeed, this setting leads to good performance in the signal recovery. In our simulation, we set the signal length $N=512$ and the number of nonzero entries $K=20$, and vary the sample size $M$ from 40 to 140 with step size of 5. The nonzero entries of the sparse signal are randomly located with the amplitudes following from a zero-mean unit-variance Gaussian distribution. We consider two matrix ensembles for the sensing matrix $\m{A}$, including Gaussian ensemble and uniform spherical ensemble (with columns uniformly distributed on the sphere $\bS^{N-1}$). To obtain the desired SNR, zero-mean AWGNs are added to the linear measurements where the noise variance is set to $\sigma^2=\sbra{K/M}10^{-\text{SNR}/10}$. We set $\text{SNR}=25\,\text{dB}$. The noisy measurements are used in the following signal recovery process. In the proposed algorithm, we set $\epsilon=0.01$ which results in both fast and accurate recovery (this will be illustrated in the next experiment). Denote $\tau_0=\sbra{M/N}\sigma^2$. We set $\tau=\theta\tau_0$ and consider six values of $\theta=10^{-4},\;10^{-2},\;10^{-1},\;1,\;10$ and $100$. Thus $\theta=1$ leads to the recommended value of $\tau$ in Subsection \ref{sec:tau_setting}. For each $M$, 100 random problems are generated and solved respectively using the proposed method with different $\tau$. The metrics are averaged results over the 100 trials.

Our simulation results are presented in Fig. \ref{Fig:choice_tau}. Fig. \ref{Fig:choice_tau_RMSE} and Fig. \ref{Fig:choice_tau_SuppSize} plot the RMSEs and support sizes of the proposed algorithm with Gaussian sensing matrices. It is shown that the recommended $\tau$ leads to approximately the smallest error with a reasonable number of measurements while the errors are almost the same when the sample size is small for different $\tau$'s. Fig. \ref{Fig:choice_tau_SuppSize} shows that the recommended $\tau$ results in the most accurate estimation of the support size in most cases. In addition, it is shown that a sparser solution is obtained if a smaller $\tau$ is used in the algorithm as expected. Almost identical performance is shown in Fig. \ref{Fig:choice_tau_USEMat_RMSE} and \ref{Fig:choice_tau_USEMat_SuppSize} by using the uniform spherical ensemble. Thus, we consider only the uniform spherical ensemble in the following experiments.

\begin{figure}
\centering
  \subfigure[]{
    \label{Fig:choice_tau_RMSE}
    \includegraphics[width=1.7in]{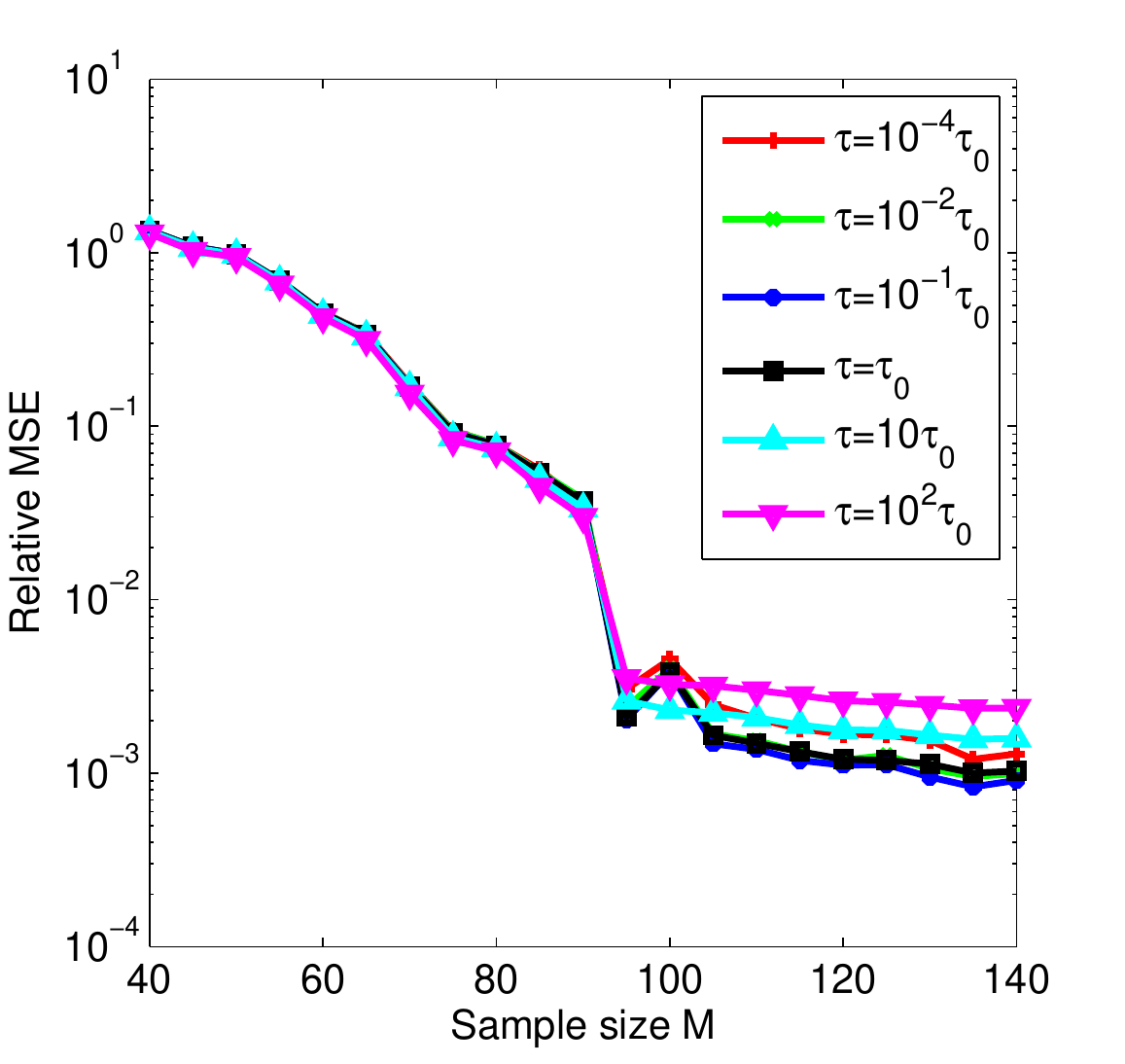}}%
  \subfigure[]{
    \label{Fig:choice_tau_SuppSize}
    \includegraphics[width=1.7in]{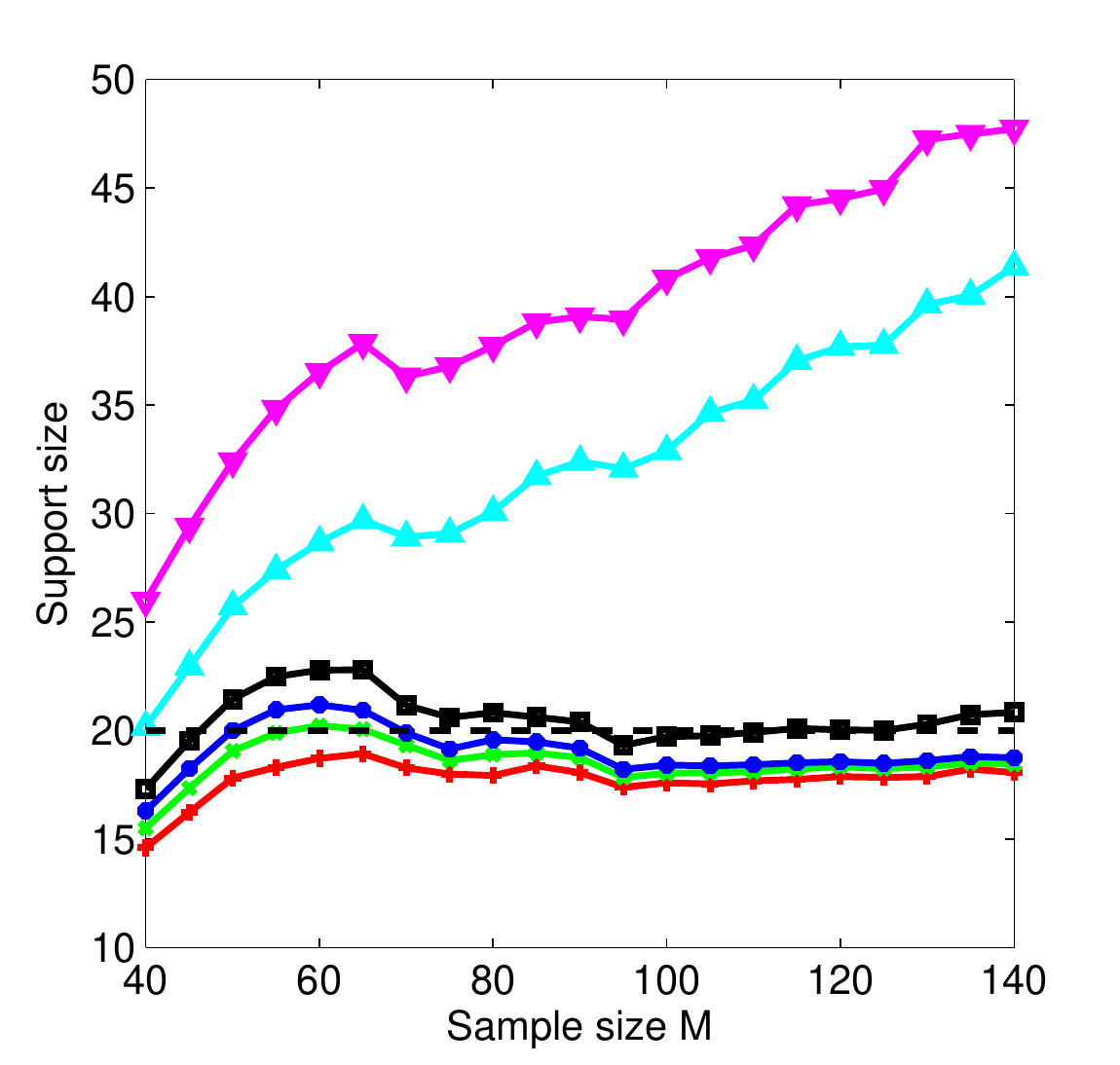}}
  \subfigure[]{
    \label{Fig:choice_tau_USEMat_RMSE}
    \includegraphics[width=1.7in]{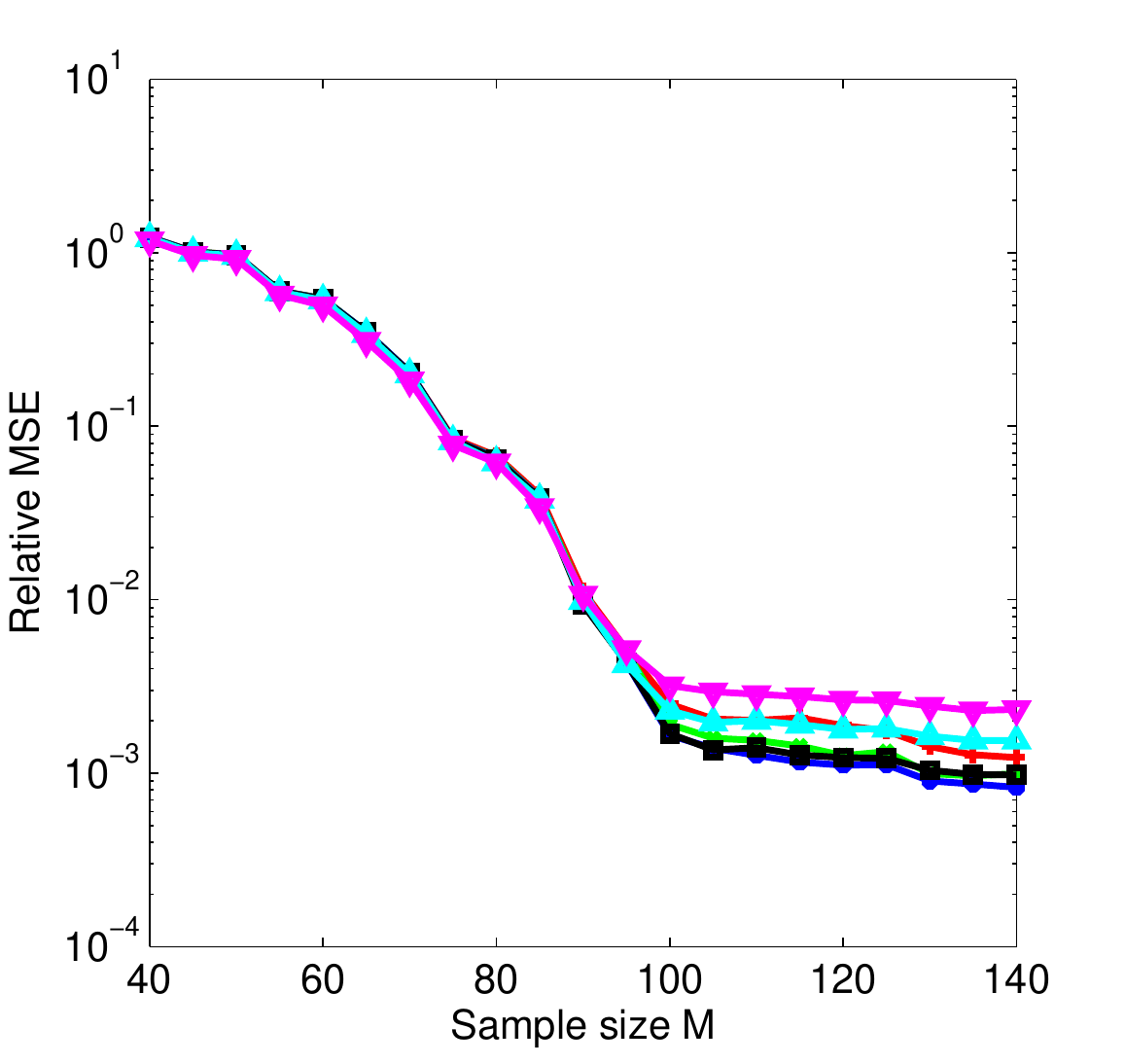}}%
  \subfigure[]{
    \label{Fig:choice_tau_USEMat_SuppSize}
    \includegraphics[width=1.7in]{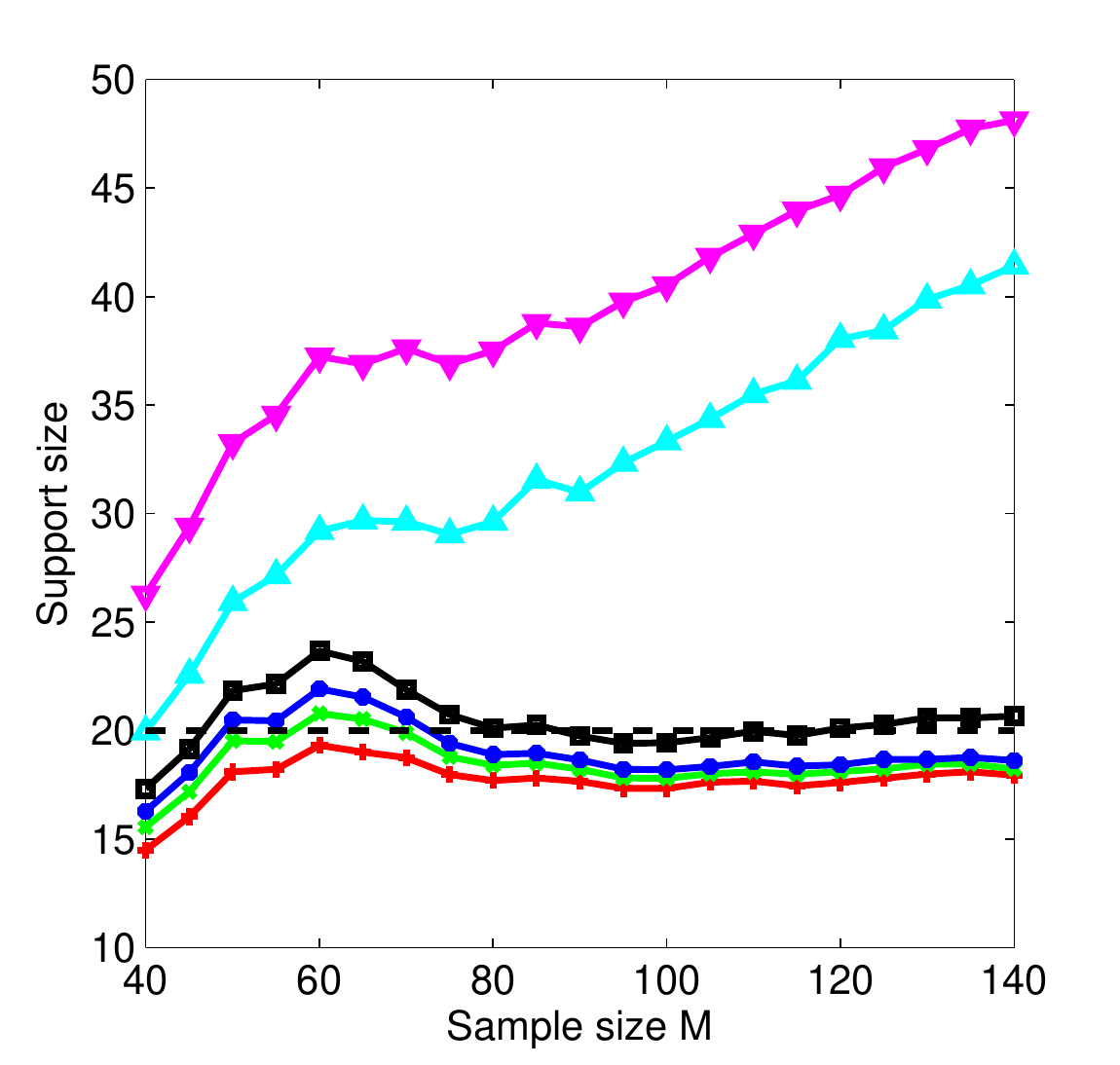}}
\centering
\caption{Performance of the proposed algorithm with respect to different settings of $\tau$ with (a) \& (b) Gaussian ensemble, and (c) \& (d) uniform spherical ensemble.} \label{Fig:choice_tau}
\end{figure}

\subsubsection{Performance with respect to $\epsilon$} We study now the performance of the proposed algorithm with respect to $\epsilon$. We repeat the simulation above using the recommended $\tau$ and consider five values of $\epsilon=0,\;0.01,\;0.1,\;0.5$ and $1$. Note that the case $\epsilon=1$ corresponds to the Laplace prior. Our simulation results are presented in Fig. \ref{Fig:performance_wrt_epsilon}. It is shown in Fig. \ref{Fig:performance_wrt_epsilon_RMSE} that the signal recovery error decays as the sample size increases in general. As the sample size is small the estimation errors differ slightly. But with a reasonable number of measurements a smaller $\epsilon$ results in a smaller error. It is shown in Fig. \ref{Fig:performance_wrt_epsilon_SuppSize} that a smaller $\epsilon$ leads to a sparser solution as expected and more accurate support size estimation. Another advantage of adopting a small $\epsilon$ can be observed in Fig. \ref{Fig:performance_wrt_epsilon_Niter} where it is shown that a smaller $\epsilon$ leads to less number of iterations. In general, the time consumption is proportional to the number of iterations since the computational workload is approximately the same at each iteration. Fig. \ref{Fig:performance_wrt_epsilon_Time} shows an exception at $\epsilon=0$ as illustrated in Remark \ref{rem:timeatepsilonis0}. In this case, the update of $\eta$ takes most of the computational time in our simulation. Since it is shown in Figs. \ref{Fig:performance_wrt_epsilon_RMSE} -- \ref{Fig:performance_wrt_epsilon_Niter} that the performance at $\epsilon=0$ and $0.01$ is hardly distinguishable, we use $\epsilon=0.01$ in the rest simulations.

\begin{figure}
\centering
  \subfigure[]{
    \label{Fig:performance_wrt_epsilon_RMSE}
    \includegraphics[width=1.7in]{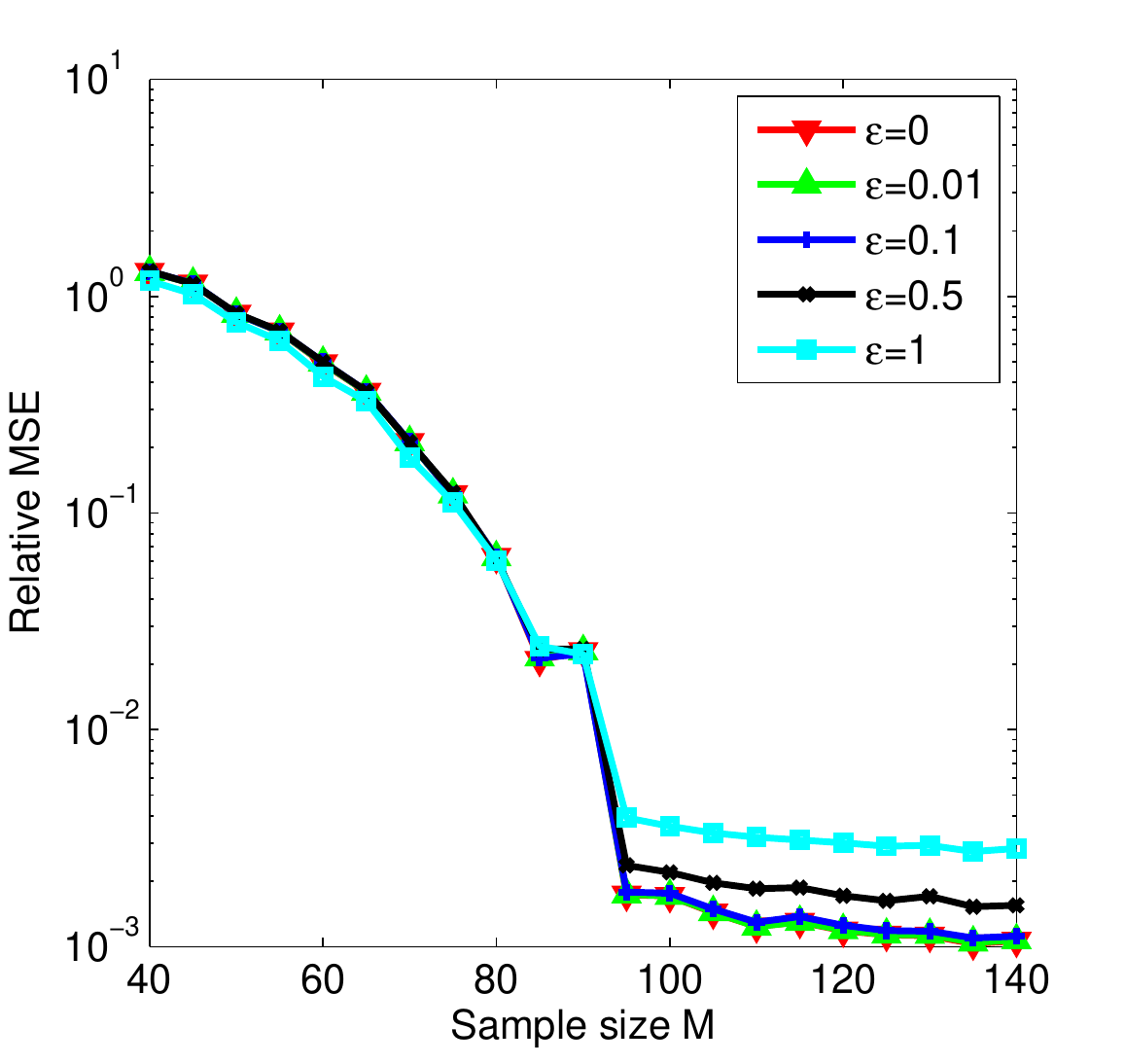}}%
  \subfigure[]{
    \label{Fig:performance_wrt_epsilon_SuppSize}
    \includegraphics[width=1.7in]{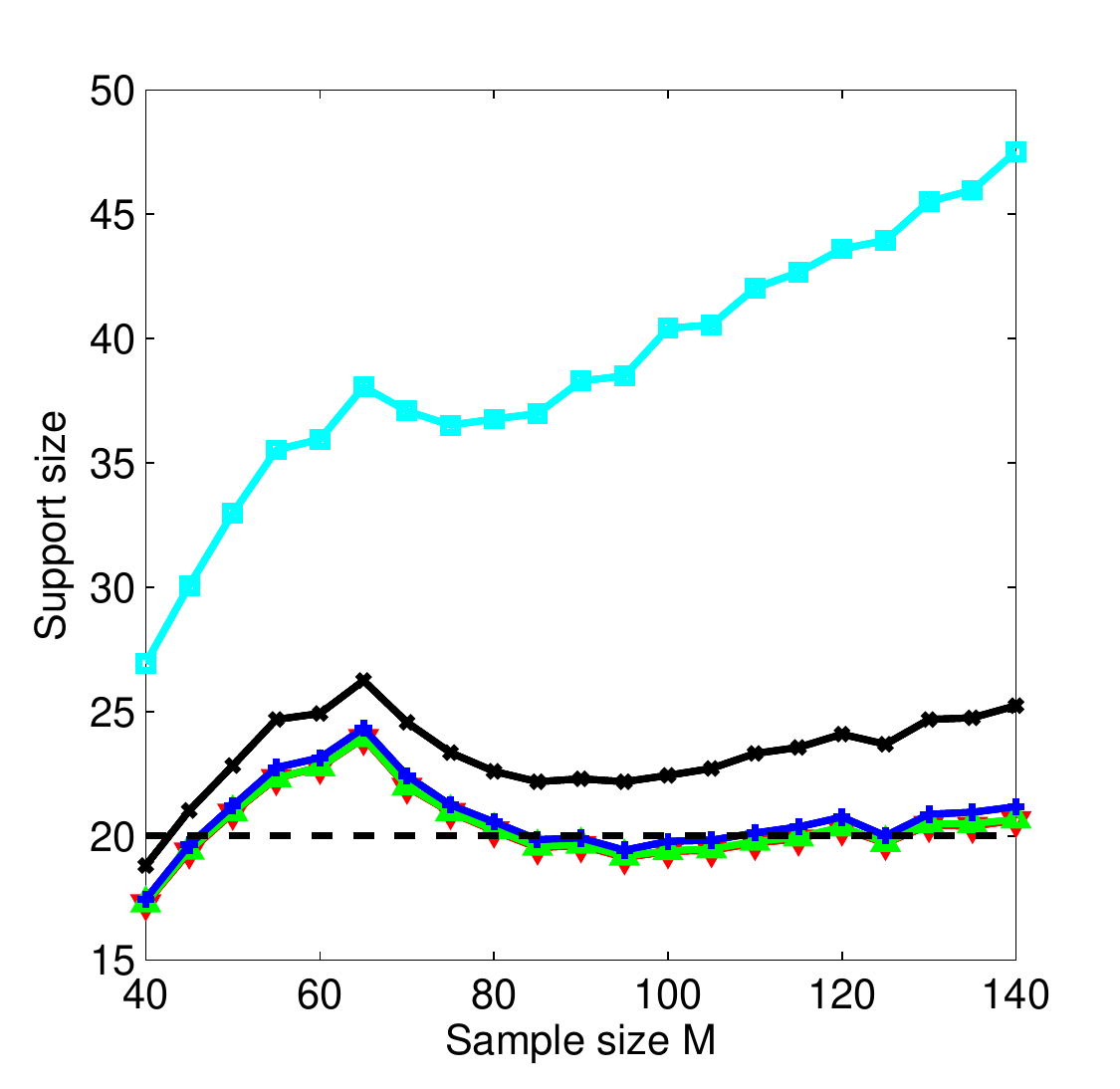}}
  \subfigure[]{
    \label{Fig:performance_wrt_epsilon_Niter}
    \includegraphics[width=1.7in]{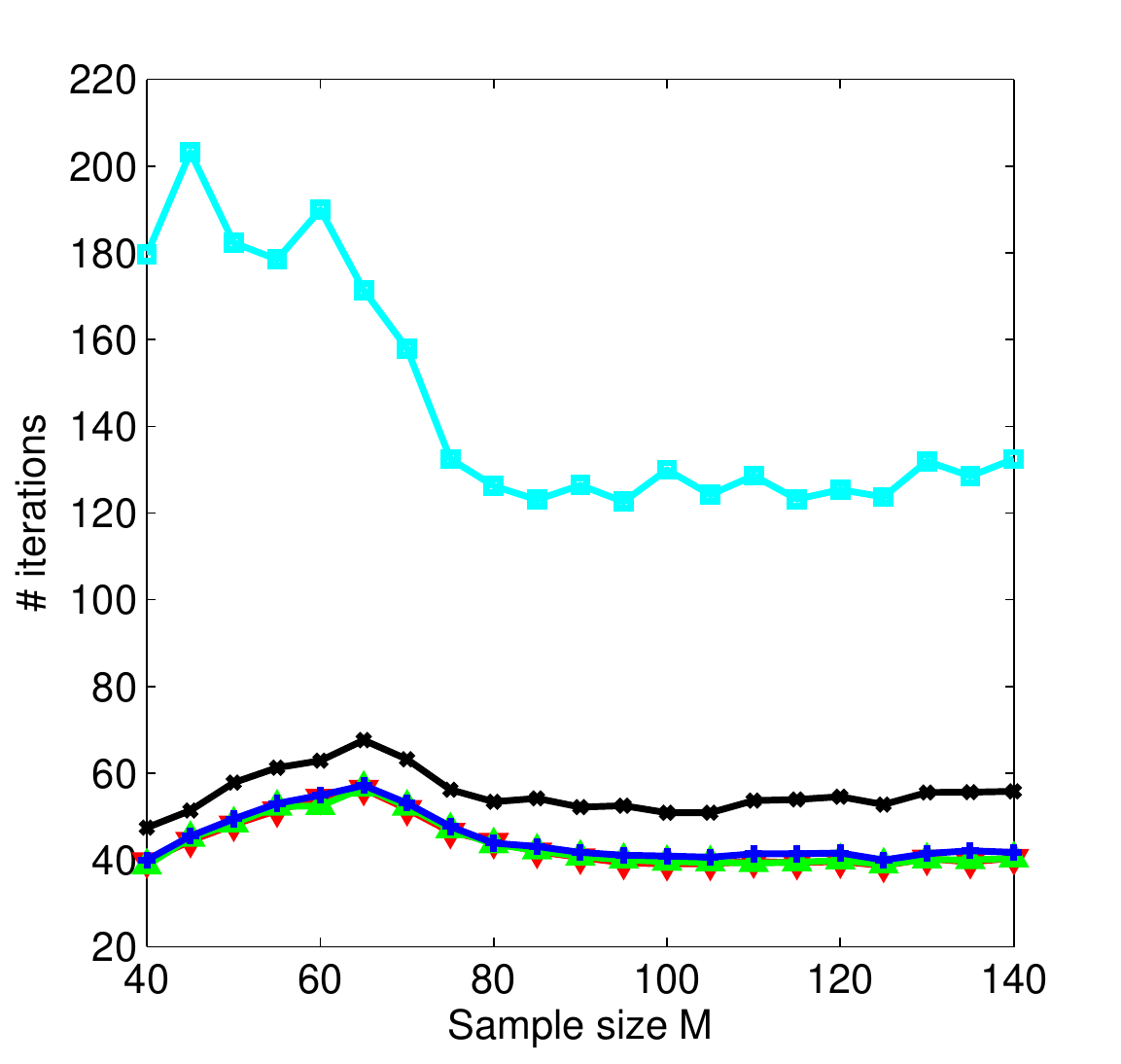}}%
  \subfigure[]{
    \label{Fig:performance_wrt_epsilon_Time}
    \includegraphics[width=1.7in]{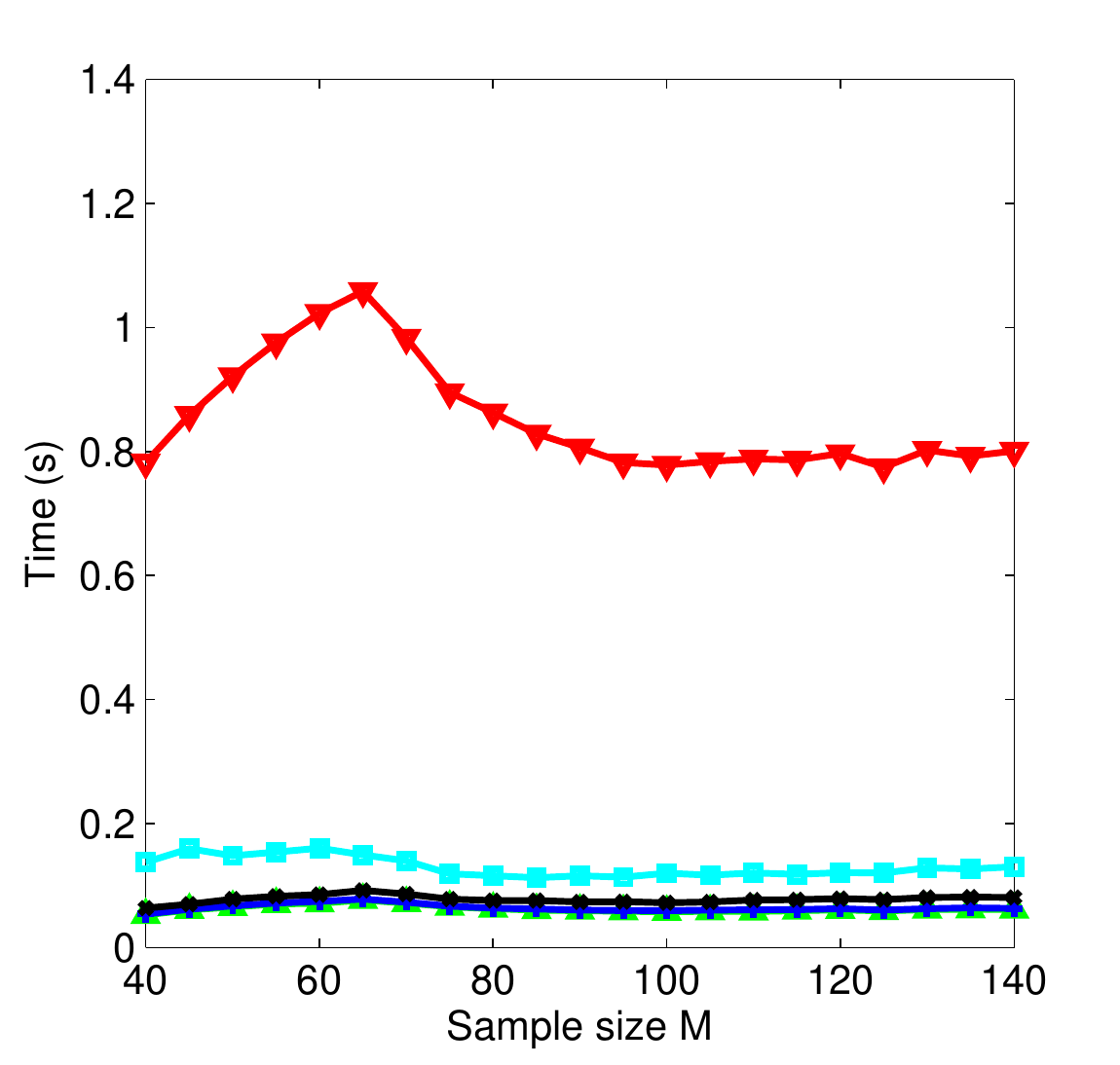}}
\centering
\caption{Performance of the proposed algorithm with respect to different settings of $\epsilon$.} \label{Fig:performance_wrt_epsilon}
\end{figure}

\subsubsection{Comparison with existing methods} We consider two simulation setups. In the first case we repeat the simulations above (i.e., we fix the $\text{SNR}=25\,\text{dB}$ and vary $M$). In each trial, all methods share the same data. We adopt the FAR thresholding strategy in StOMP. Our simulation results are presented in Fig. \ref{Fig:comparison_case1}. It is shown in Fig. \ref{Fig:comparison_case1_RMSE} that the reconstruction errors of the three SBL methods (BCS, Laplace and our proposed method) are very close to each other and larger than those of BPDN, RW-BPDN and StOMP if the sample size is small. With a reasonable sample size it can be seen that our proposed method has the smallest error. Fig. \ref{Fig:comparison_case1_SuppSize} shows the average support size of the recovered signal. The results of BPDN and RW-BPDN are omitted since they are global optimization methods and their numerical solutions have no exact zero entries. In general, the estimated support sizes of StOMP, BCS and Laplace increase with the sample size. As expected the proposed method produces sparser solutions than BCS and Laplace. It is shown that the proposed method can accurately estimate the support of the sparse signal in most cases and has the best performance. Fig. \ref{Fig:comparison_case1_Niter} plots the number of iterations of the three SBL methods, where it is shown that the proposed one uses the least number of iterations and thus is the fastest one in computational speed. On average, StOMP uses the least computation time (about $0.01$s), followed by the SBL methods (from $0.06$ to $0.1$s), and then BPDN (about $1$s) and RW-BPDN (about $2$s). We note that a number of solvers have been proposed to solve the BPDN problem with improved speed, e.g., SPGL1\cite{van2008probing}.

\begin{figure*}
\centering
  \subfigure[]{
    \label{Fig:comparison_case1_RMSE}
    \includegraphics[width=2.33in]{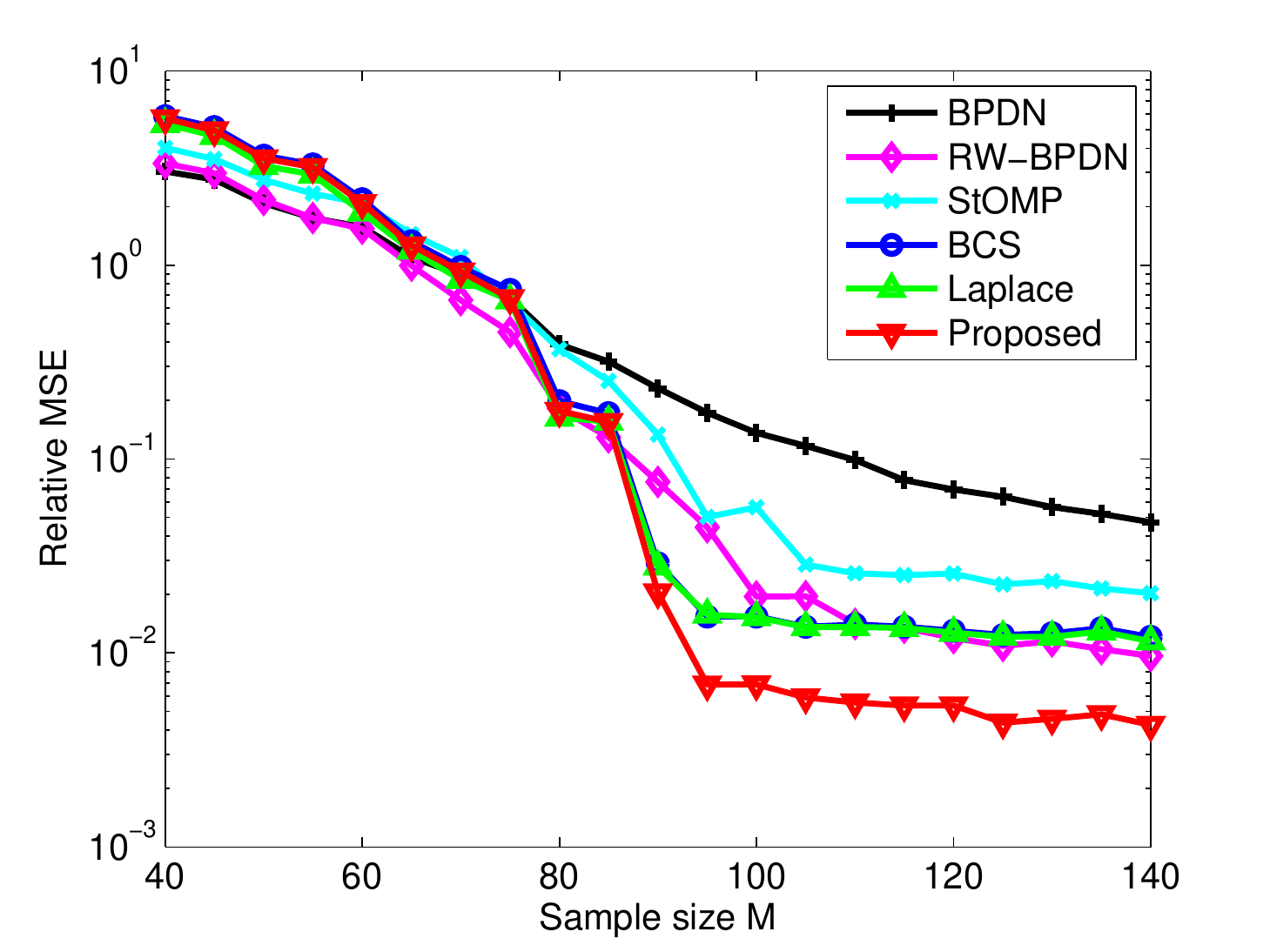}}%
  \subfigure[]{
    \label{Fig:comparison_case1_SuppSize}
    \includegraphics[width=2.33in]{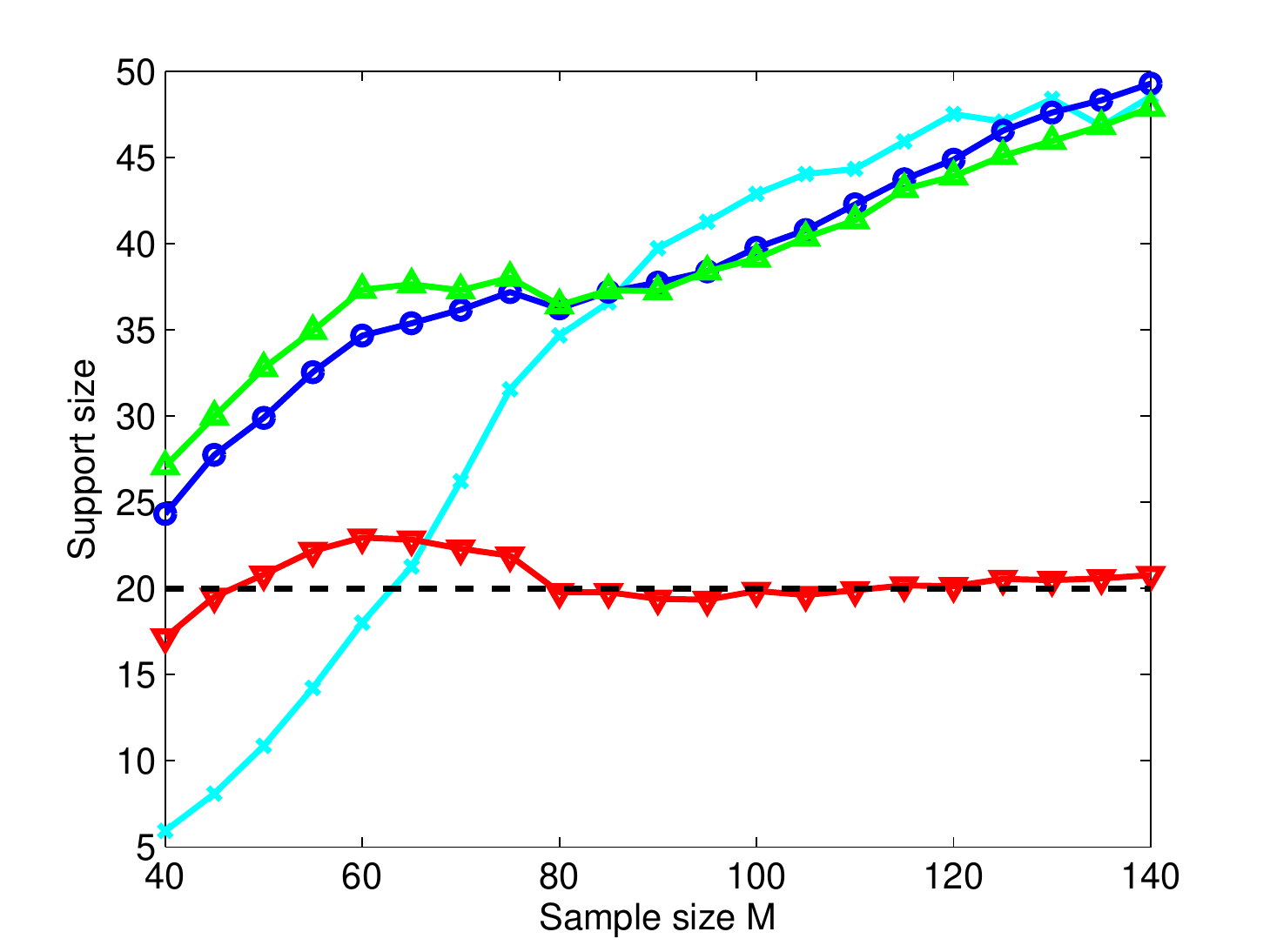}}%
  \subfigure[]{
    \label{Fig:comparison_case1_Niter}
    \includegraphics[width=2.33in]{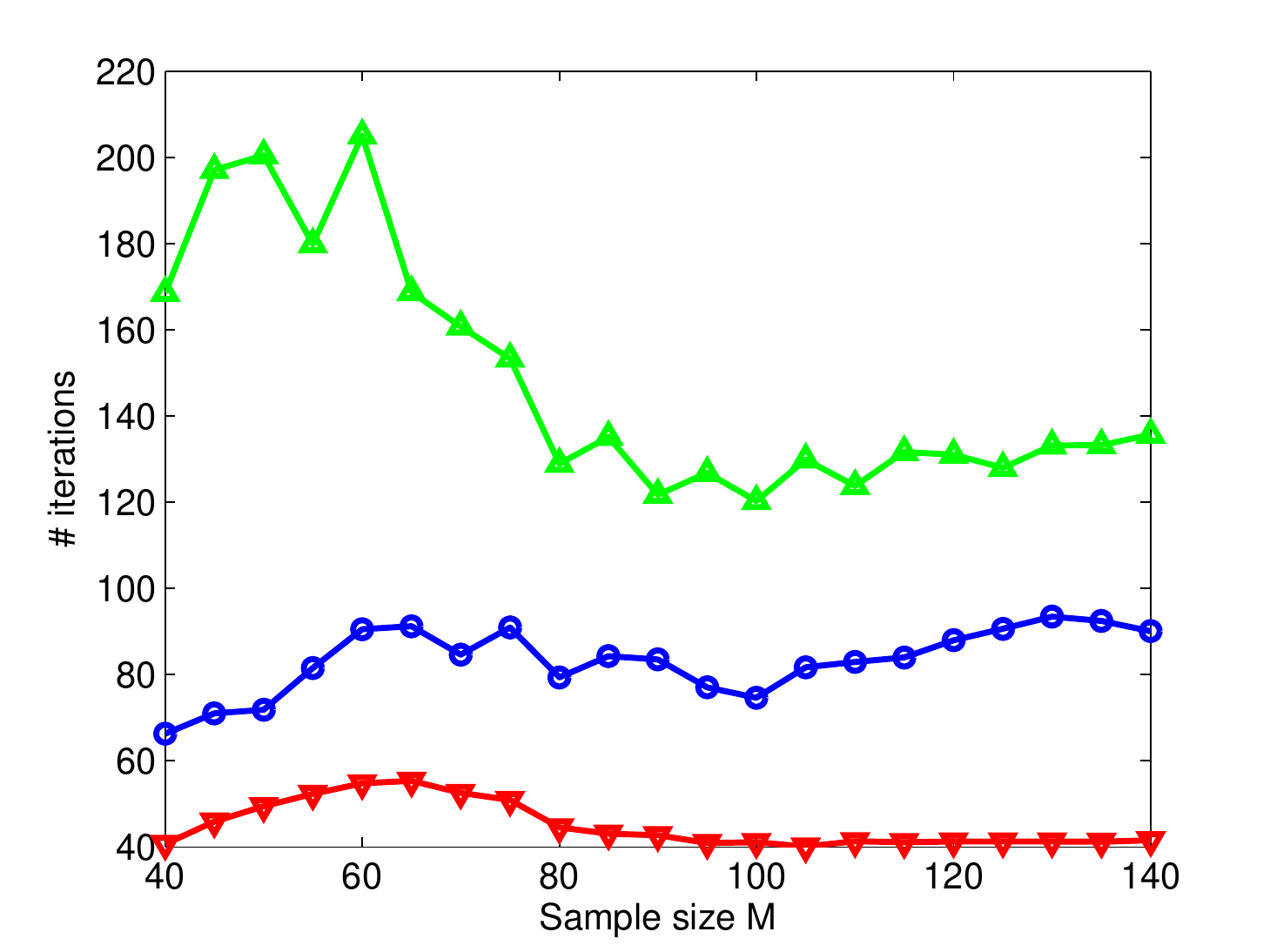}}
\centering
\caption{Performance comparison of the proposed algorithm with existing ones with $\sbra{N, K, \text{SNR}}=\sbra{512, 20, 25 \text{dB}}$.} \label{Fig:comparison_case1}
\end{figure*}

In the next simulation we set the sample size $M=120$ and vary the SNR from 0 to 50dB with step size of 5dB. The simulation results are presented in Fig. \ref{Fig:comparison_case2}. Fig. \ref{Fig:comparison_case2_RMSE} shows that the proposed method has consistently the smallest signal recovery error. Fig. \ref{Fig:comparison_case2_SuppSize} shows that the proposed method produces the sparsest solution and the most accurate support size estimation. Fig. \ref{Fig:comparison_case2_Niter} shows that among the three SBL methods the proposed one uses the least number of iterations at all SNR levels. In the low SNR regime, it can be 6 and 3 times less in comparison with Laplace and BCS respectively, leading to that the proposed method is much faster than Laplace and BCS.

\begin{figure*}
\centering
  \subfigure[]{
    \label{Fig:comparison_case2_RMSE}
    \includegraphics[width=2.33in]{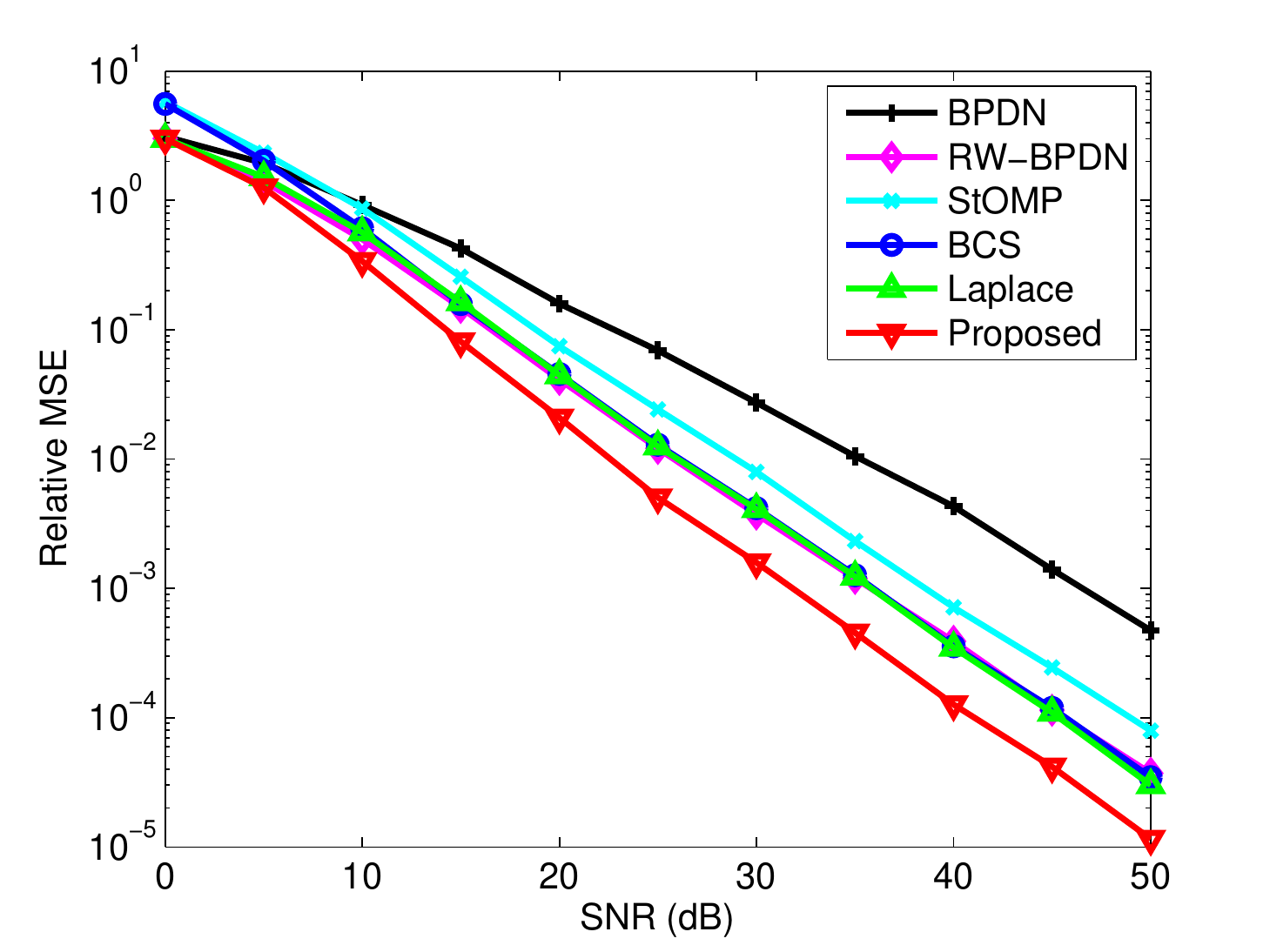}}%
  \subfigure[]{
    \label{Fig:comparison_case2_SuppSize}
    \includegraphics[width=2.33in]{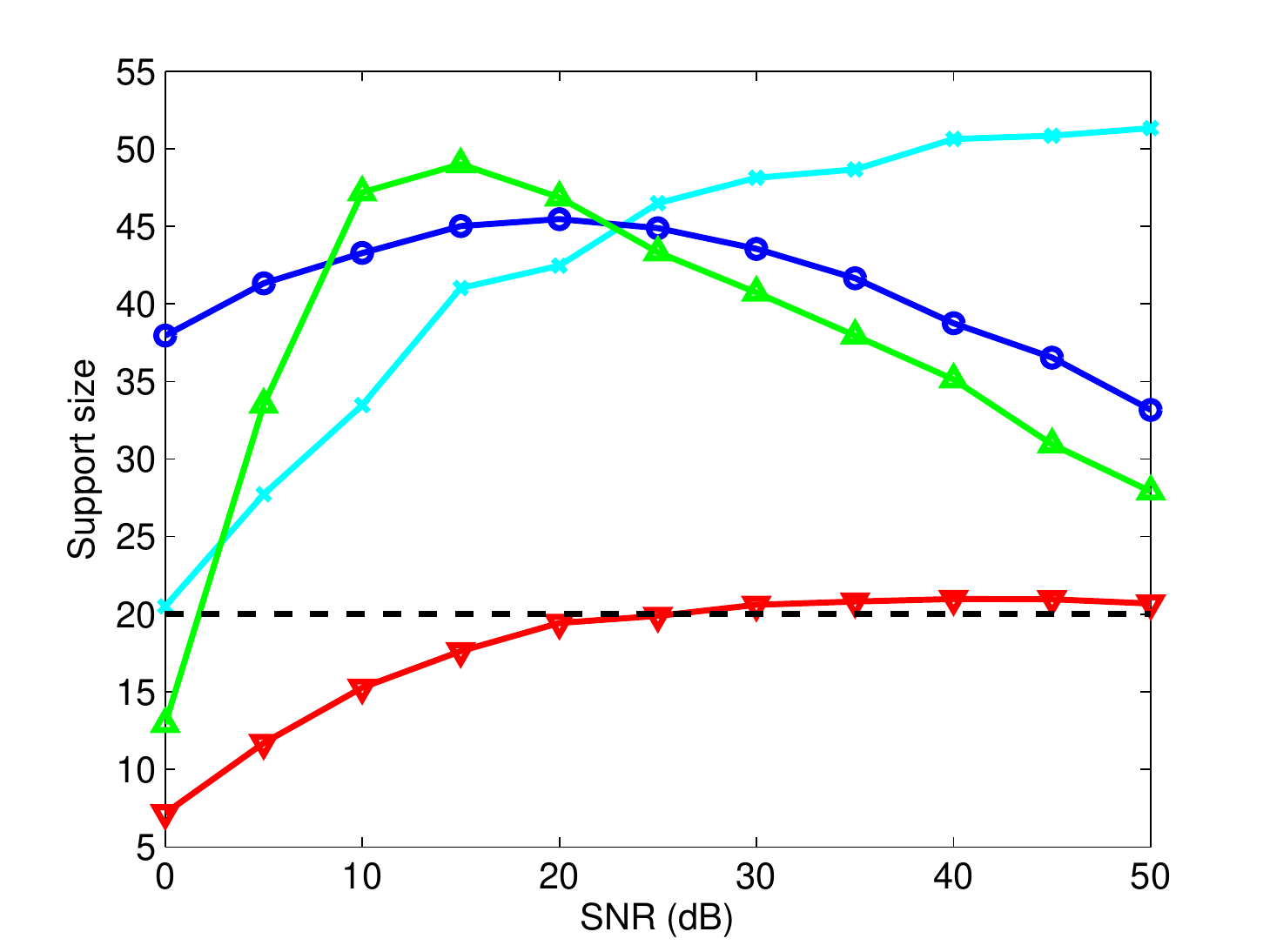}}%
  \subfigure[]{
    \label{Fig:comparison_case2_Niter}
    \includegraphics[width=2.33in]{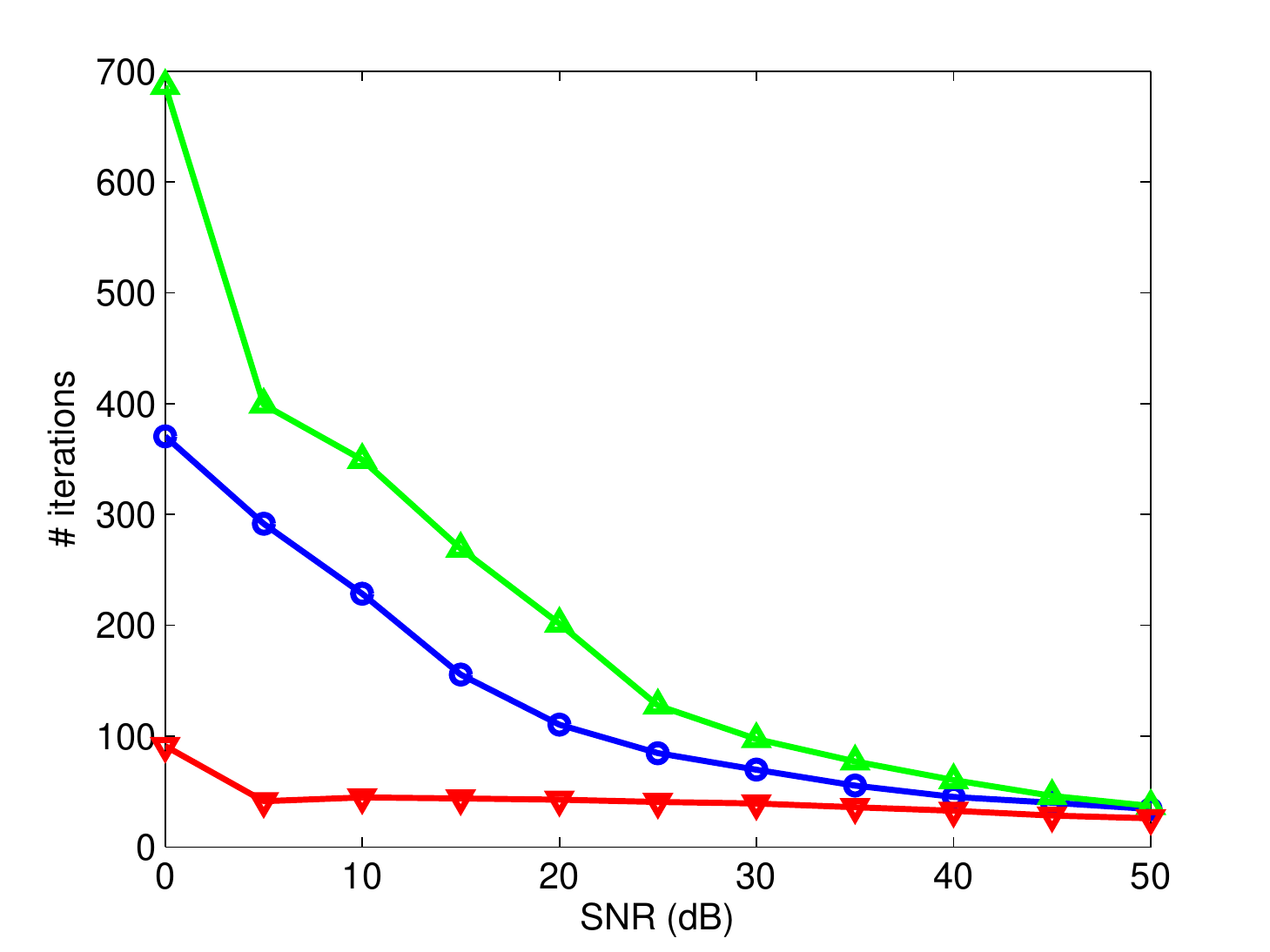}}
\centering
\caption{Performance comparison of the proposed algorithm with existing ones with $\sbra{N, M, K}=\sbra{512, 120, 20}$.} \label{Fig:comparison_case2}
\end{figure*}

In summary, the proposed method has improved performance for sparse signals in comparison with existing ones. It outperforms its SBL peers in both signal recovery accuracy and computational speed.

\subsection{Images}
In this section, we revisit the widely used multiscale CS reconstruction \cite{tsaig2006extensions} of the $512\times512$ Mondrian image in SparseLab. We use the same simulation setup, i.e., we choose the ``symmlet8'' wavelet as the sparsifying basis with a coarsest scale $j_0=4$ and a finest scale $j_1=6$. The number of wavelet samples is $N=4096$ and the sample size of CS methods is $M=2713$. The parameters of BP and StOMP with the FDR and FAR thresholding strategies (denoted by FDR and FAR respectively) are set as in SparseLab. Since the wavelet expansion of the Mondrian image is compressible but not exactly sparse, we set $\sigma^2=0.01Var\sbra{\m{y}}$ in Laplace and our proposed method as in BCS, where $Var\sbra{\m{y}}$ denotes the variance of the entries of $\m{y}$.

Table \ref{table:Time} presents the experimental results over 100 trials. Linear reconstruction from $4096$ wavelet samples has a reconstruction error of $0.1333$ that represents a lower bound of the error of the considered CS methods. The global optimization method BP has the smallest error among the CS methods, followed by BCS, Laplace, the proposed method and StOMP. The presented results verify again that the proposed method produces a sparser solution than BCS and Laplace. In fact, it produces the sparsest solution among all the methods. So it is reasonable that the proposed method has a slightly worse reconstruction error in comparison with BCS and Laplace since the original signal is not exactly sparse. We note that the proposed method is faster than BCS and Laplace. FDR uses the least time but has the worst accuracy. In comparison with FAR, the proposed method is slightly slower but more accurate. Finally, it can be observed that the proposed method has the most stable performance among the CS methods except BP by comparing the standard deviation of the metrics. We note that BP can be accelerated using recently developed algorithms for $\ell_1$ optimization. Fig. \ref{Fig:image} shows examples of reconstructed images where faithful reconstructions of the Mondrian image can be observed.

\begin{table}
 \caption{Averaged Relative MSEs, CPU Times and Number of Nonzero Entries ($\text{mean}\pm\text{standard deviation}$) for Multiscale CS Reconstruction of the Mondrian Image. }
 \centering
\begin{tabular}{l|l|l|l}
  \hline\hline
   & RMSE & Time (s) & \# Nonzeros \\\hline
   Linear& $0.1333$ & --- & $4096$ \\
   BP&$0.1393\pm0.0008$ & $42.2\pm4.03$ & $4096\pm0$ \\
   FDR& $0.1999\pm0.0487$ & $8.84\pm2.12$ & $2155\pm122$ \\
   FAR& $0.1499\pm0.0033$ & $17.0\pm4.35$ & $1142\pm41$ \\
   BCS& $0.1423\pm0.0023$ & $27.2\pm5.92$ & $1305\pm67$ \\
   Laplace& $0.1429\pm0.0013$ & $25.7\pm6.11$ & $1218\pm65$ \\
   Proposed & $0.1448\pm0.0011$ & $21.3\pm4.09$ & $1049\pm21$ \\
  \hline\hline
\end{tabular}\label{table:Time}
\end{table}

\begin{figure*}
\centering
  \subfigure[Mondrian]{
    \label{Fig:Mondrian}
    \includegraphics[width=1.7in]{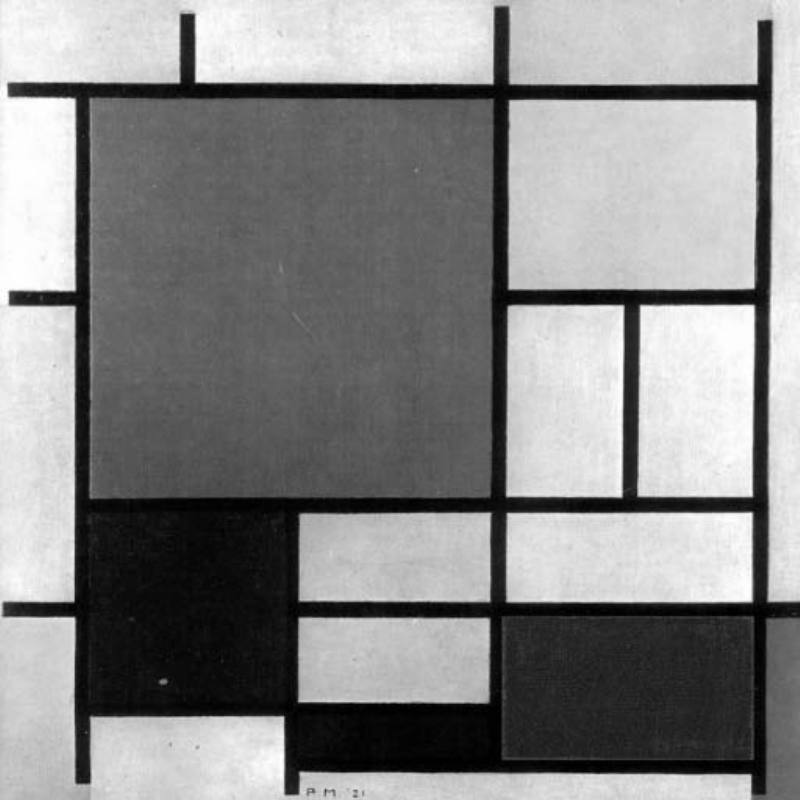}}%
  \subfigure[Linear]{
    \label{Fig:recon_LIN}
    \includegraphics[width=1.7in]{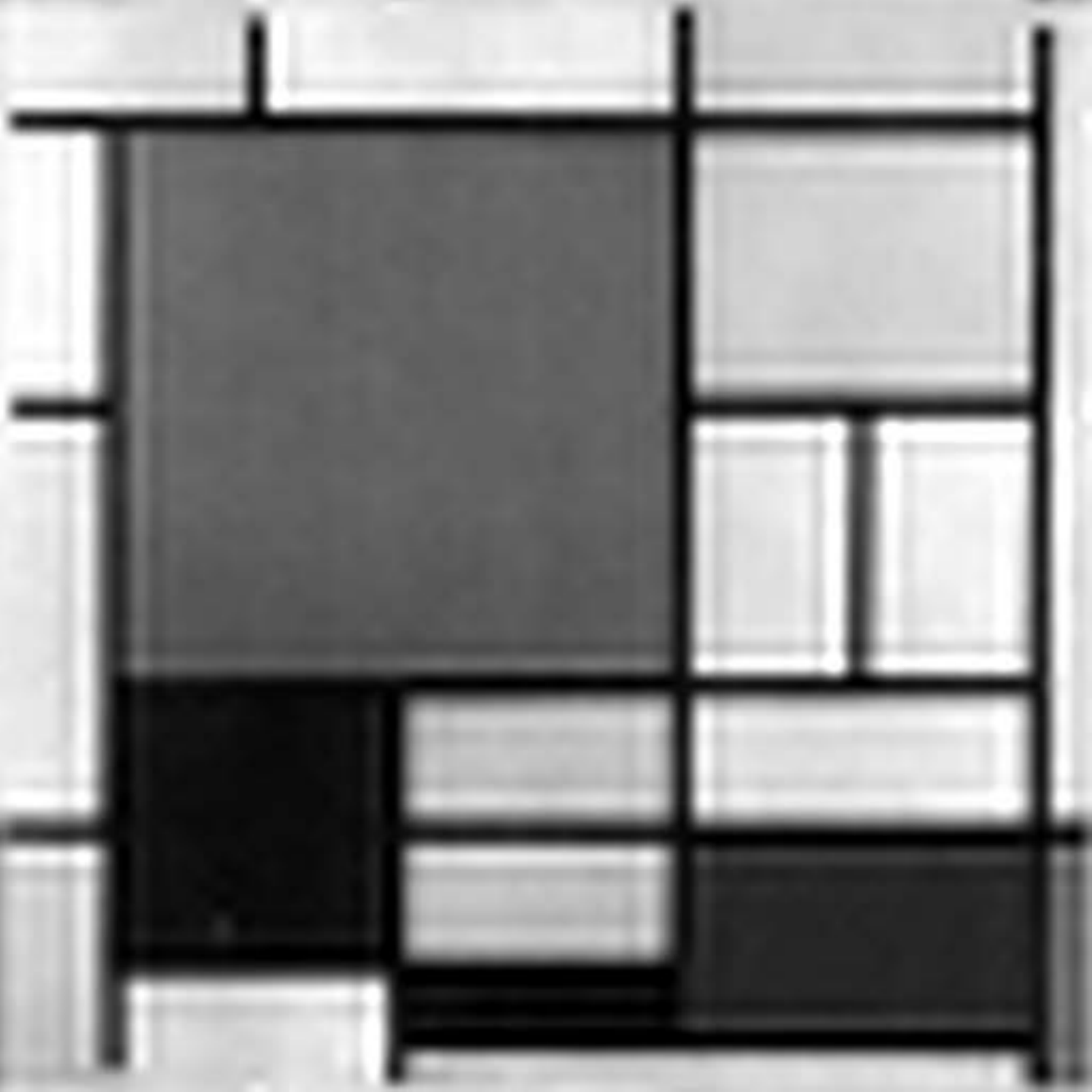}}%
  \subfigure[BP]{
    \label{Fig:recon_BP}
    \includegraphics[width=1.7in]{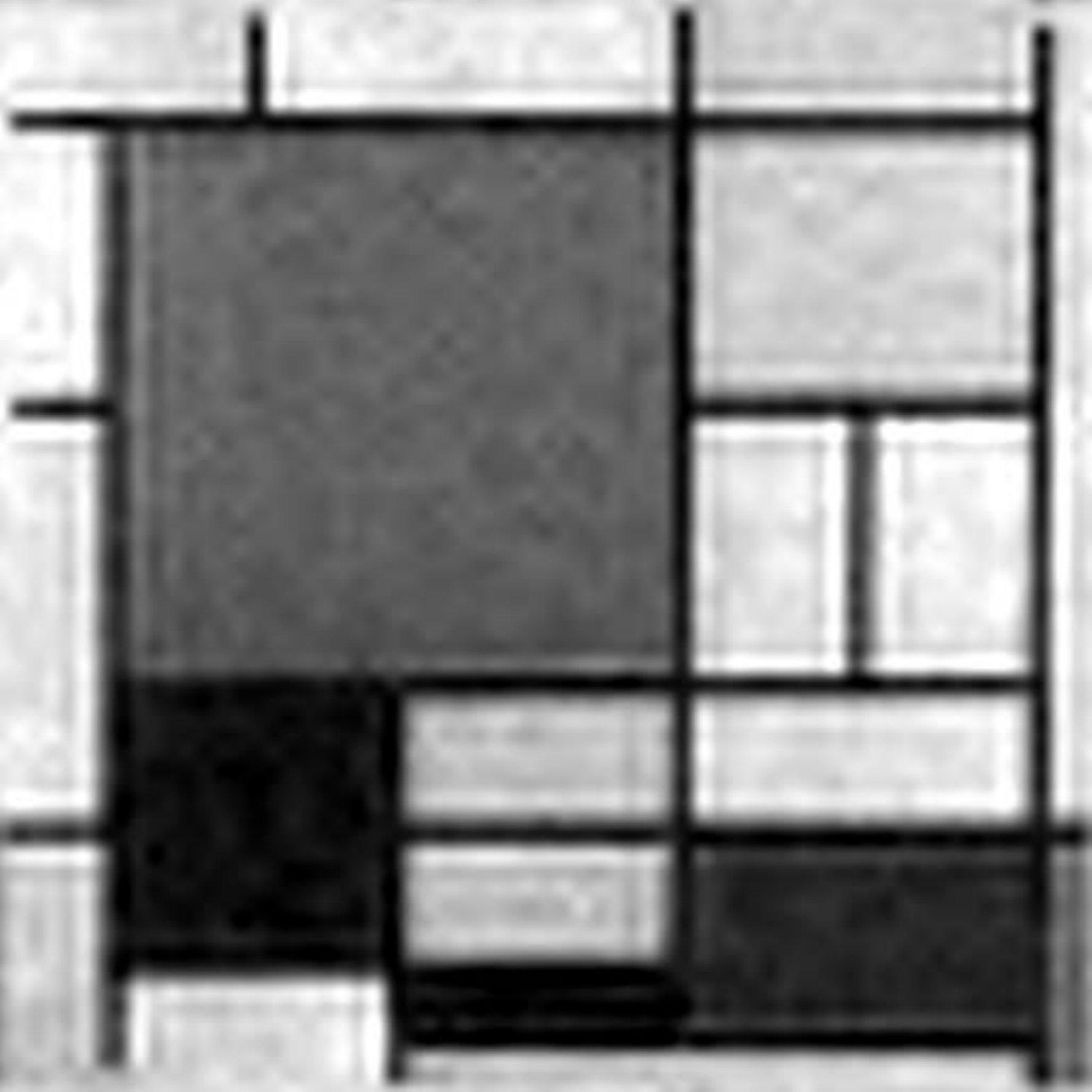}}%
  \subfigure[FDR]{
    \label{Fig:recon_FDR}
    \includegraphics[width=1.7in]{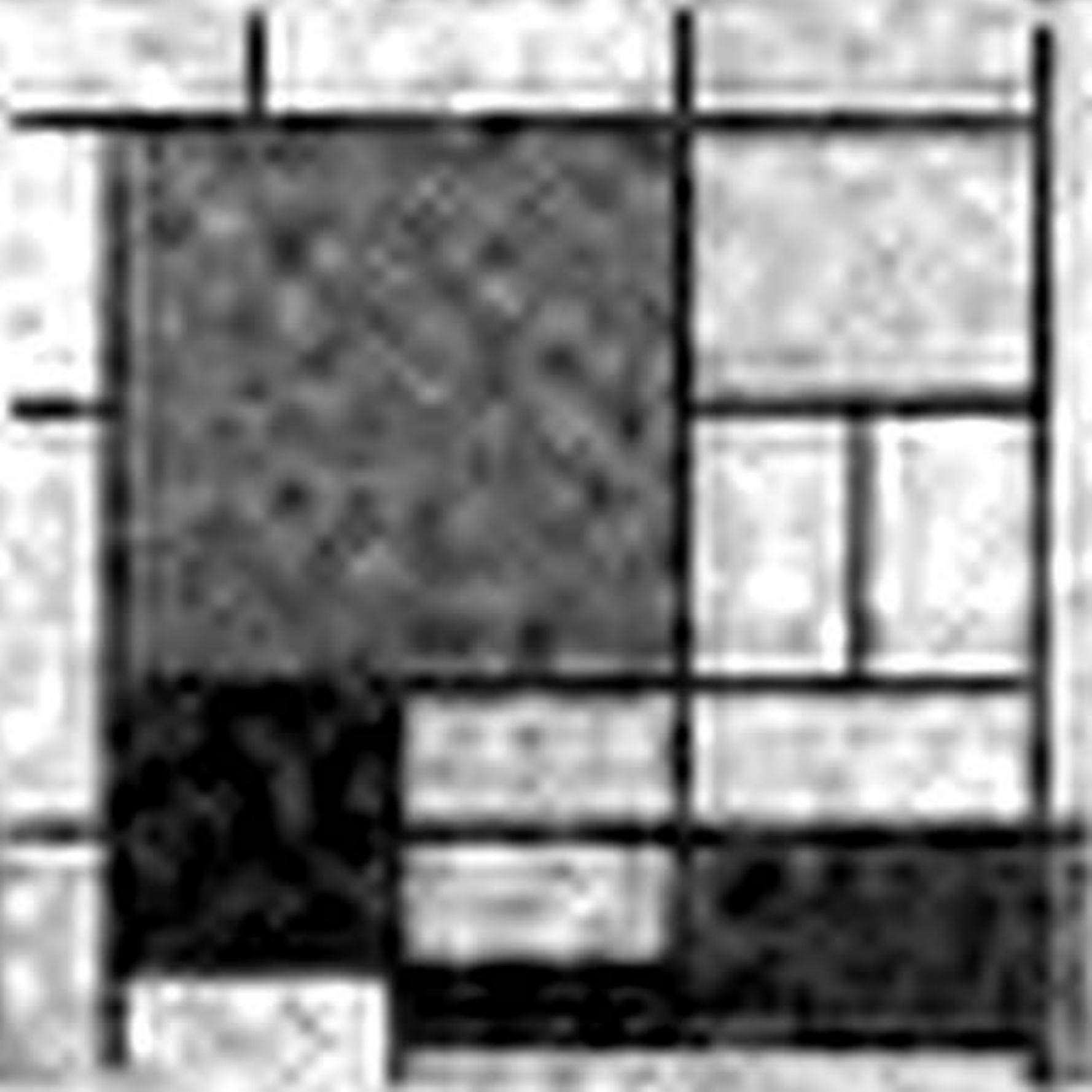}}
  \subfigure[FAR]{
    \label{Fig:recon_FAR}
    \includegraphics[width=1.7in]{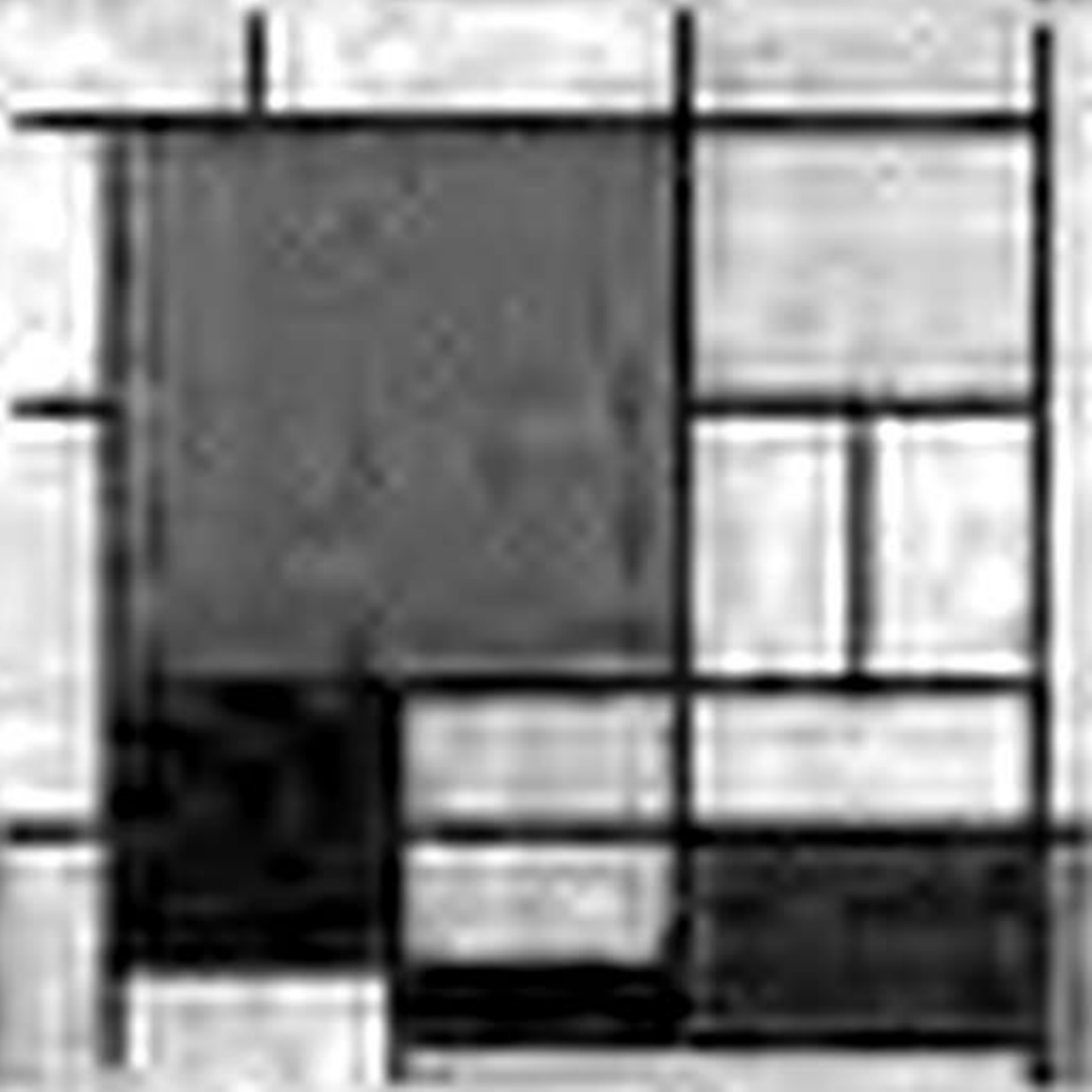}}%
  \subfigure[BCS]{
    \label{Fig:recon_BCS}
    \includegraphics[width=1.7in]{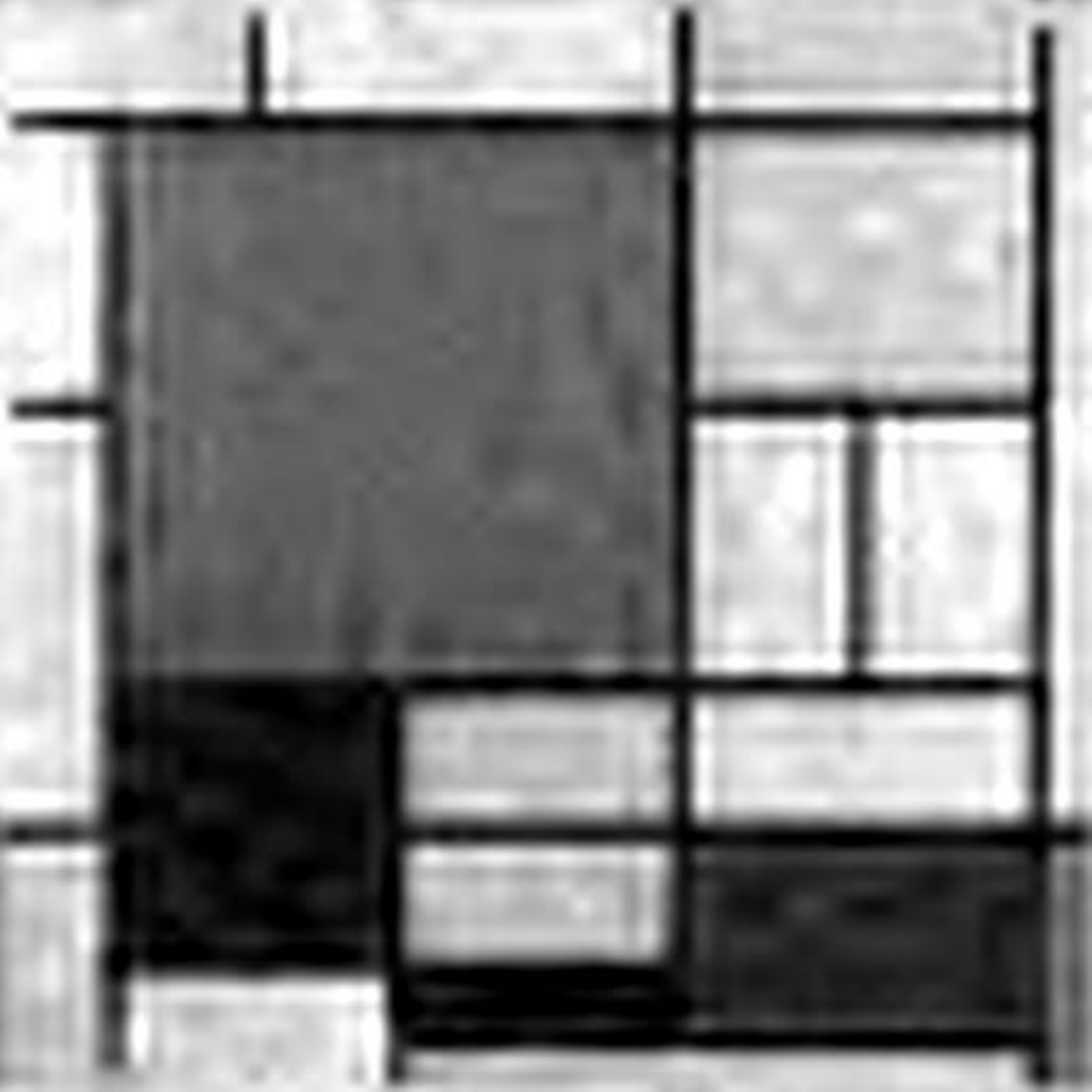}}%
  \subfigure[Laplace]{
    \label{Fig:recon_Lap}
    \includegraphics[width=1.7in]{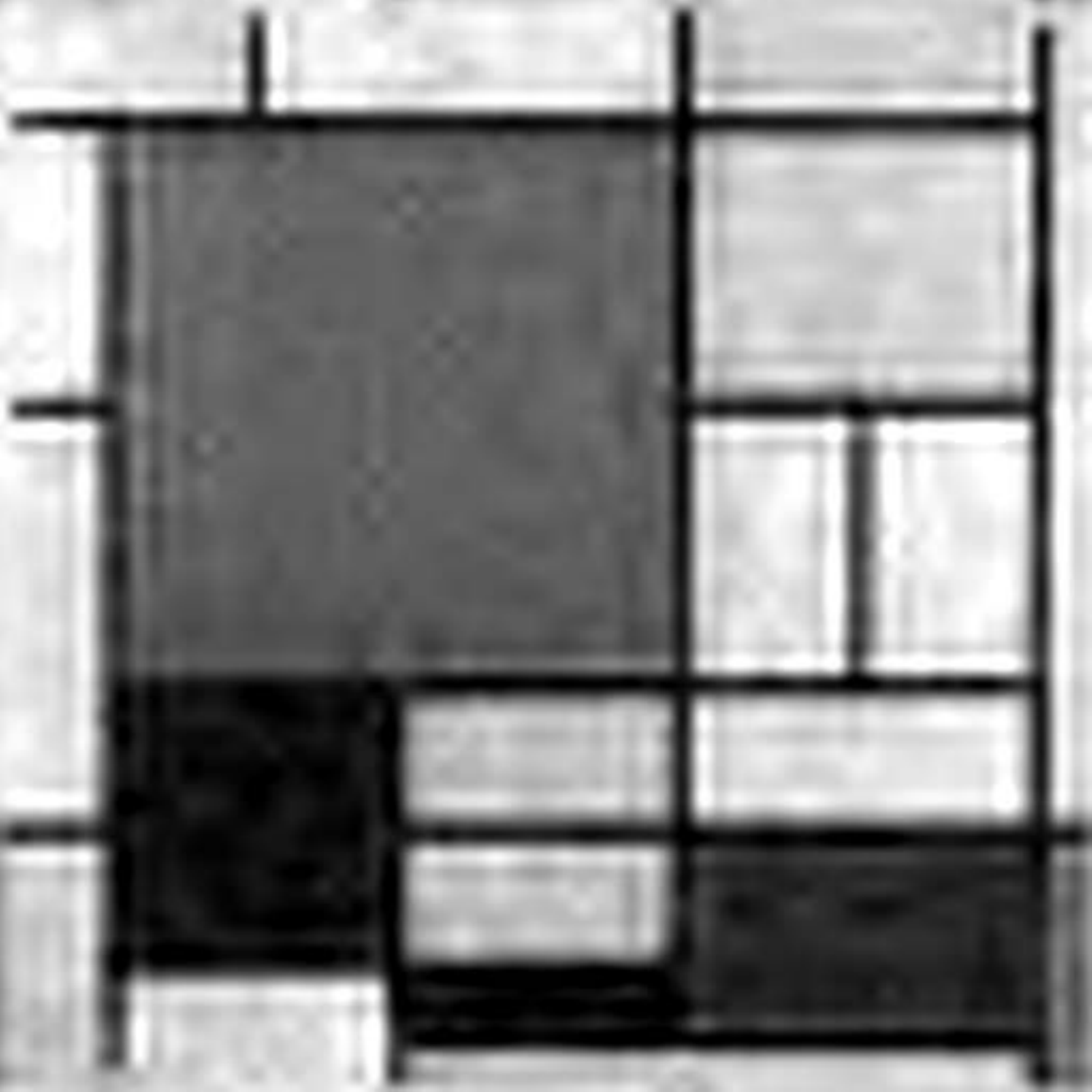}}%
  \subfigure[Proposed]{
    \label{Fig:recon_yang}
    \includegraphics[width=1.7in]{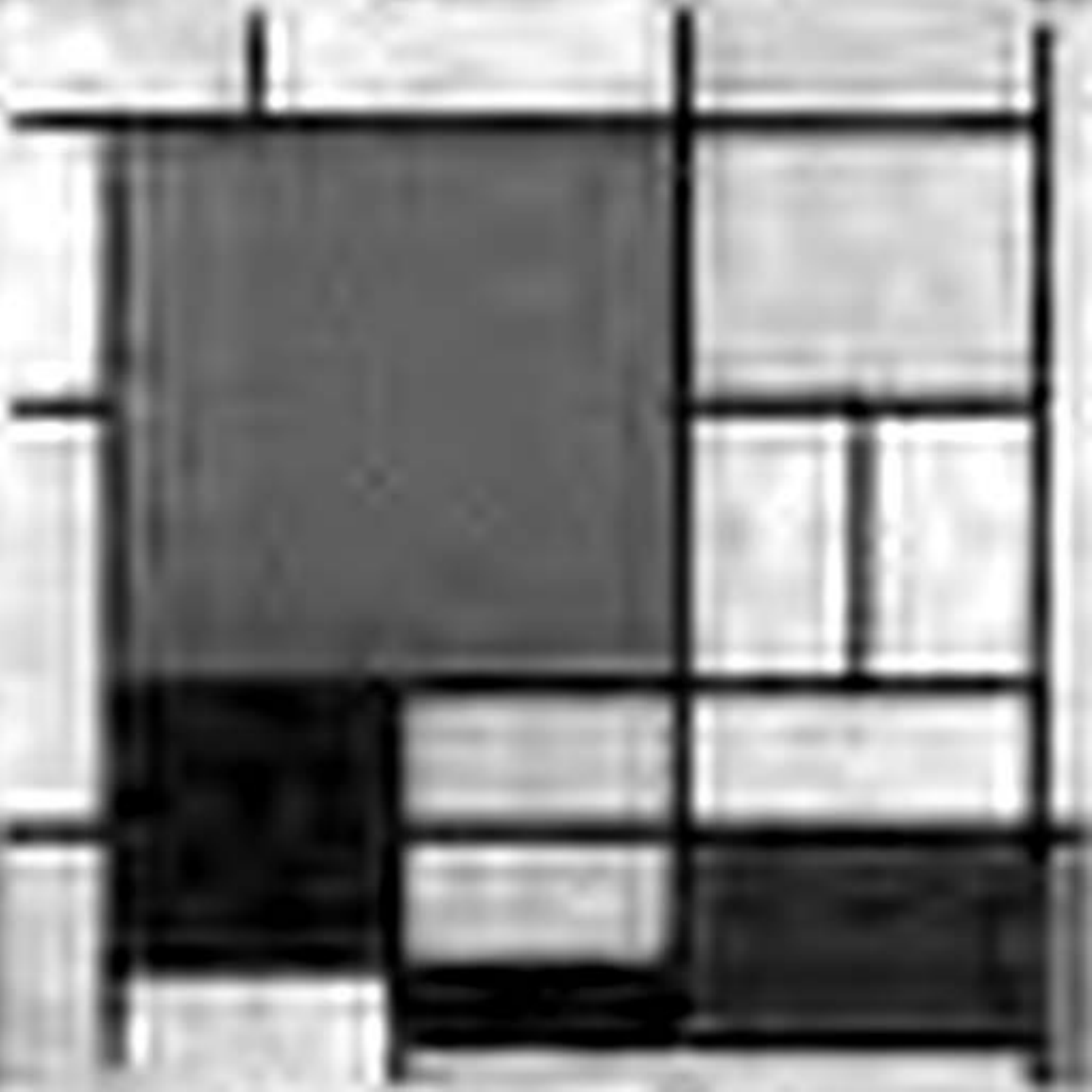}}
\centering
\caption{The $512\times512$ Mondrian image (a) and its reconstructions using (b) linear reconstruction ($\text{RMSE}=0.1333$) from $N=4096$ wavelet samples and a multiscale CS scheme from $M=2713$ linear measurements by (c) BP ($\text{RMSE}=0.1391$, $\text{time}=44.7s$ and $\text{\# nonzeros}=4096$), (d) StOMP with FDR thresholding ($\text{RMSE}=0.1751$, $\text{time}=10.3s$ and $\text{\# nonzeros}=2014$), (e) StOMP with FAR thresholding ($\text{RMSE}=0.1529$, $\text{time}=15.7s$ and $\text{\# nonzeros}=1088$), (f) BCS ($\text{RMSE}=0.1448$, $\text{time}=24.1s$ and $\text{\# nonzeros}=1293$), (g) Laplace ($\text{RMSE}=0.1427$, $\text{time}=25.5s$ and $\text{\# nonzeros}=1229$) and (h) our proposed method ($\text{RMSE}=0.1440$, $\text{time}=19.8s$ and $\text{\# nonzeros}=1033$).} \label{Fig:image}
\end{figure*}

\section{Conclusion} \label{sec:conclusion}
The sparse signal recovery problem in CS was studied in this paper. Within the framework of Bayesian CS, a new hierarchical sparsity-inducing prior was introduced and efficient signal recovery algorithms were developed. Similar theoretical results on the global and local optimizations of the proposed method  were proven as that for the basic SBL. The main algorithm was shown to produce sparser solutions than its existing SBL peers. Numerical simulations were carried out to demonstrate the improved performance of the proposed sparsity-inducing prior and solution. The proposed G-STG prior preserves the general structure of existing hierarchical sparsity-inducing priors and can be implemented in other SBL-based methods with ease.

\bibliographystyle{IEEEtran}


\end{document}